\documentclass[aps,prx,longbibliography,twocolumn,nofootinbib]{revtex4-2}

\usepackage[utf8]{inputenc}
\usepackage[T1]{fontenc}
\usepackage{amsmath}
\usepackage{amssymb}
\usepackage{amsthm}
\usepackage{mathtools}
\usepackage{enumitem}
\usepackage{hyperref}
\usepackage{cleveref}
\usepackage{booktabs}
\usepackage{graphicx}
\usepackage{xcolor}
\usepackage{threeparttable}

\crefname{equation}{Eq.}{Eqs.}
\Crefname{equation}{Equation}{Equations}

\hypersetup{
  colorlinks=true,
  linkcolor=blue,
  citecolor=blue,
  urlcolor=blue
}

\newtheorem{theorem}{Theorem}
\newtheorem{lemma}[theorem]{Lemma}

\newcommand{\tr}[2][]{\operatorname{Tr}_{#1}\!\left[#2\right]}

\newcommand{\Fock}{\mathcal{F}}
\newcommand{\SN}[2]{\mathcal{S}_{#1,#2}}

\newcommand{\rank}{\operatorname{rank}}
\newcommand{\ket}[1]{\left|#1\right\rangle}
\newcommand{\bra}[1]{\left\langle#1\right|}
\newcommand{\braket}[2]{\left\langle#1\middle|#2\right\rangle}
\newcommand{\ketbra}[2]{\left|#1\middle\rangle\middle\langle#2\right|}

\begin{document}

\title{Orbital-Rotation Shadow Tomography Reduces the Classical Cost of Quantum-Classical Auxiliary-Field Quantum Monte Carlo}

\author{Luning Zhao}
\affiliation{IonQ Inc, College Park, MD, 20740, USA}
\email[Corresponding author: ]{zhao@ionq.co}

\author{Joshua Goings}
\affiliation{IonQ Inc, College Park, MD, 20740, USA}

\author{Evgeny Epifanovsky}
\affiliation{IonQ Inc, College Park, MD, 20740, USA}

\author{Martin Roetteler}
\affiliation{IonQ Inc, College Park, MD, 20740, USA}

\begin{abstract}
Quantum-classical auxiliary-field quantum Monte Carlo (QC-AFQMC) is a promising route to using quantum computers for strongly correlated electronic structure problems in chemistry, but its scalability is limited by a classical post-processing bottleneck: the measured quantum trial state must be used to evaluate overlaps, force biases, and local energies for many walkers throughout the AFQMC propagation. Existing matchgate-shadow approaches avoid the exponential cost of generic classical shadows, but their post-processing remains computationally expensive for large molecular systems. In this work, we show that particle-number-preserving fermionic shadows, implemented as orbital-rotation shadows, are able to drastically reduce QC-AFQMC post-processing costs. We derive efficient estimators for overlaps and their first and second derivatives, giving $\mathcal{O}(N^3)$ scaling for overlaps and force biases and $\mathcal{O}(N^4)$ scaling for local energies per shadow and walker at fixed filling and with $\mathcal{O}(N)$ Cholesky vectors, matching the asymptotic scaling of classical single-Slater-determinant AFQMC. For fixed normalized target determinants, we also prove that the complex-overlap estimator is unbiased with variance upper bounded by $\mathcal{O}(\eta)$, where $\eta$ is the particle number, yielding a sufficient budget of $\mathcal{O}(\eta/\epsilon^2)$ ideal independent measurements for additive root-mean-square overlap error $\epsilon$. We combine these estimators with a novel virtual-correlation-energy treatment and a GPU-accelerated implementation, validate the method against exact overlap calculations, and demonstrate QC-AFQMC using orbital-rotation shadows. Performance projections for $\pi$-conjugated molecules up to pentacene with a (22e, 22o) active space and 356 total orbitals show that the full classical post-processing workflow can be completed in slightly more than two days on 8$\times$A100 graphics processing units (GPUs). These results significantly reduce classical post-processing costs for QC-AFQMC with quantum trial states and motivate further improvements in preparation and measurement of the trial state on quantum hardware.
\end{abstract}

\maketitle

\section{Introduction}
\label{sec:introduction}

One of the most promising applications of quantum computers is the modeling of strongly correlated electronic systems in molecules and materials~\cite{cao2019quantum,mcardle2020quantum,bauer2020quantum,Patel25_7490,Yuri25_111335}. Examples include catalysts with multiple transition-metal centers such as the FeMo cofactor in nitrogenase~\cite{reiher2017elucidating}, large $\pi$-conjugated organic molecules such as acenes~\cite{anthony2006acenes,hachmann2007radical}, electronically excited states, and, potentially, high-temperature superconductors~\cite{lee2006doping,keimer2015quantum}.

Several approaches have been developed in response to this opportunity.
Motivated by formally provable exponential quantum speedups relative to exact classical simulation methods, the quantum phase estimation (QPE) algorithm~\cite{abrams1999quantum,aspuru2005simulated,omalley2016scalable,reiher2017elucidating, Neven18_041015, cao2019quantum, Babbush21_030305} computes the energy and other properties of a molecular system on fault-tolerant quantum computing devices. However, current demonstrations of QPE remain limited to small model systems because of hardware constraints.

For this reason, near-term algorithms such as variational quantum algorithms (VQAs)~\cite{tilly2022variational, peruzzo2014variational} have emerged as an alternative suitable for noisy quantum hardware. VQAs include the variational quantum eigensolver, hardware-efficient and unitary-coupled-cluster ansatz families~\cite{Zhao22_60, kandala2017hardware, Gambetta17_23879}, adaptive ansatz constructions~\cite{grimsley2019adaptive}, and related variational optimization frameworks~\cite{peruzzo2014variational,farhi2014quantum,kandala2017hardware,romero2018strategies,hempel2018quantum,grimsley2019adaptive,tilly2022variational,cerezo2021variational}.  Although VQAs can be executed on near-term quantum devices, their effectiveness in practice suffers from barren plateaus, prohibitive number of measurements~\cite{gonthier2022measurements} and sensitivity to hardware noise~\cite{mcclean2018barren,cerezo2021variational,gonthier2022measurements}. 

Two alternative hybrid quantum-classical approaches leverage the capabilities of near-term hardware to sample from prepared quantum states followed up by classical post-processing. The first approach, sample-based quantum diagonalization (SQD), including closely related quantum-selected configuration-interaction methods, uses quantum samples to identify important many-electron configurations and then build and diagonalize projected Hamiltonians with a classical high-performance computer~\cite{kanno2023quantum,robledo2024chemistry,yu2025quantum}. In demonstrations on near-term quantum computers, these algorithms have exhibited noise resilience and delivered highly accurate results for large systems~\cite{Motta26_01138} although questions remain about the benefit of SQD over state-of-the-art conventional selected configuration-interaction methods~\cite[Secs.~4.3--4.4]{reinholdt2025critical}.

This paper focuses on the second of those hybrid approaches, the quantum-classical auxiliary-field quantum Monte Carlo (QC-AFQMC) method~\cite{huggins2022unbiasing,huang2024evaluating,zhao2026matchgate}. QC-AFQMC is a quantum extension of classical AFQMC~\cite{motta2018abinitio, lee2022twenty, Zhao20_174105}, which propagates a stochastic representation of the wave function in imaginary time and in the long-time limit converges the wave function to the ground state~\cite{zhang1997constrained,zhang2003phaseless,motta2018abinitio,lee2022twenty,shee2023transformative}.
In practice, however, a trial state is used to constrain the propagation and control the sign problem, which is a fundamental bottleneck for generic projection based fermionic Monte Carlo methods~\cite{Umrigar143_164105, troyer2005complexity}. This constraint introduces an energy bias that vanishes when the trial state is the exact eigenstate. Quantum computer's ability to prepare highly accurate and nearly exact trial state could significantly benefit AFQMC to reduce and even completely remove the bias. This was first demonstrated in a 2022 study~\cite{huggins2022unbiasing}, where a trial state was prepared on a quantum computer and used in AFQMC to guide imaginary time propagation. Preliminary results on molecular and lattice systems have shown that quantum trial states can significantly reduce the energy bias compared with classical trial states making QC-AFQMC a promising algorithm for quantum simulation.

Each time step of QC-AFQMC involves calculating the overlap between a trial state $\ket{\Psi_T}$ and a non-orthogonal Slater determinant $\ket{\phi}$. For a quantum trial state, the naive approach is to invoke the Hadamard test to evaluate the overlap, however the need for frequent queries to the quantum computer renders this idea impractical. Instead, the classical shadow technique~\cite{aaronson2020shadow,huang2020predicting,elben2023randomized} estimates the properties of $\ket{\Psi_T}$, including the necessary overlap $\braket{\Psi_T}{\phi}$, by converting randomized-basis measurements of $\ket{\Psi_T}$ into desired observables. It was this approach that enabled the first demonstration of QC-AFQMC on quantum hardware.

Although classical shadows provide a way to estimate the overlaps, whether this can be achieved efficiently (in polynomial time) depends on the structure of the measurement basis. For example, the first demonstration~\cite{huggins2022unbiasing} of QC-AFQMC used Clifford shadows~\cite{huang2020predicting}, which require an exponentially scaling runtime for overlap evaluations thus severely limiting the application of QC-AFQMC. Matchgate shadows~\cite{Miyake24_57,wan2023matchgate} have addressed this problem by applying a design from fermionic operators, which preserve fermionic symmetry. Wan \textit{et al.}~\cite{wan2023matchgate} gave an explicit finite-size variance bound for complex-overlap estimation and an $\mathcal{O}(N^4)$ overlap algorithm per shadow and target Slater determinant; their numerical results suggest sublinear variance growth. The first end-to-end demonstration of matchgate-shadow QC-AFQMC was limited to the hydrogen molecule and a spin defect model with 4 qubits~\cite{huang2024evaluating}. Since then, the applicability has been extended to chemically relevant molecular systems~\cite{zhao2026matchgate,nappi2026qcafqmc}.

Despite these successes, matchgate-shadow QC-AFQMC still has significant practical limitations. One notable issue is the cost of shadow post-processing in AFQMC, including the calculation of force biases and local energies. In earlier derivative-based implementations, the arithmetic costs per shadow and walker are $\mathcal{O}(N^4)$ for overlaps and force biases and $\mathcal{O}(N^5)$ for local energies~\cite[Sec.~II.1, discussion after Eq.~(4); Sec.~IV.3, discussion before Eq.~(22)]{nappi2026qcafqmc}. 
Nappi \textit{et al.}~\cite[Eqs.~(4), (8), (36)--(38), and the accompanying complexity discussion]{nappi2026qcafqmc} subsequently reduced these per-shadow and per-walker costs to $\mathcal{O}(N^3)$ for overlaps, $\mathcal{O}(N^3+N^2N_{\gamma})$ for force biases, and $\mathcal{O}(N^3N_{\gamma})$ for local energies, where $N_{\gamma}$ is the number of Cholesky vectors. For $N_{\gamma}=\mathcal{O}(N)$, the force-bias cost is also cubic.

Given these advances in the framework of fermionic symmetry preserving matchgate shadows, it becomes natural to ask whether it is possible to further bring down the cost by designing fermionic shadow tomography that preserves more than just the fermionic symmetry. One natural choice is particle-number symmetry since molecular electronic-structure calculations commonly target a fixed particle number. If one applies the constraints to matchgate unitaries so that they need to commute with number operators, one obtains a special set of unitaries that are of the form of single-particle orbital rotation unitaries in quantum chemistry. This was first studied by Zhao \textit{et al.}~\cite{Zhao21_110504}, and then more thoroughly by Low~ \cite{low24_08964}, who derived an efficient estimator of the orbital rotation shadows to compute $k$-body reduced density matrices ($k$-RDMs). Furthermore, the sample complexity of orbital-rotation shadow tomography   is independent of the number of qubits and that the approach admits a simpler post-processing algorithm compared to that for matchgate shadows~\cite{Larocca26_27254}.

This work builds upon Low's orbital-rotation shadow framework~\cite{low24_08964}. Starting from the shadow estimator, we derive an efficient algorithm for estimating the overlap between an arbitrary fixed-particle-number quantum state and a Slater determinant along with derivative formulas for evaluating force biases and local energies. For a fixed normalized target determinant under ideal independent sampling, we establish unbiased complex-overlap estimation with variance at most $4\eta+3$, where $\eta$ is the particle number, giving a sufficient measurement budget of $\mathcal{O}(\eta/\epsilon^2)$ for additive root-mean-square overlap error $\epsilon$. The resulting algorithm scales as $\mathcal{O}(N^3)$ for overlaps and force biases and as $\mathcal{O}(N^4)$ for local energies per shadow and walker at fixed filling and with $\mathcal{O}(N)$ Cholesky vectors matching the scaling of classical AFQMC with a single-determinant trial state. We show that besides post-processing, orbital rotation shadows also lead to a simpler and more efficient algorithm for virtual correlation energy (VCE), a crucial technique that allows QC-AFQMC to account for dynamic electron correlation outside the active space. Our runtime projections indicate that this approach yields orders-of-magnitude efficiency improvements in classical QC-AFQMC post-processing relative to our earlier matchgate-shadow implementation. This work substantially reduces the post-processing cost of QC-AFQMC and brings QC-AFQMC closer to practical large-scale applications.

The paper is organized as follows. In Sec.~\ref{sec:theory}, we derive the efficient orbital-rotation-shadow overlap estimator and present the associated algorithms and cost analysis. Section~\ref{sec:results} reports numerical validation, QC-AFQMC applications, and runtime projections. We then discuss time to solution separating quantum execution from classical post-processing. The paper concludes with a discussion of open challenges and future directions.

\section{Theory}
\label{sec:theory}

Here we derive the unbiased one-shadow determinant estimator in \cref{eq:direct_overlap_single_shot}, together with derivative formulas that can be evaluated efficiently within QC-AFQMC. We proceed as follows. The target overlap is first rewritten as an off-diagonal coherence of an augmented fixed-particle state in \cref{eq:overlap_identity}. The orbital-rotation-shadow calculation then determines the required normalization and produces the one-shot estimator. Its dependence on the target determinant and measurement record is compressed into the $\eta\times\eta$ measured-row block in \cref{eq:K_block_Phi}, which is implemented in Algorithm~1. For QC-AFQMC, the conjugate estimator takes the determinant form in \cref{eq:D_single_shadow}; differentiating this determinant gives the first and mixed second derivative formulas in \cref{eq:first_derivative,eq:second_derivative}. Table~\ref{tab:core_notation} collects the notation reused across these steps. 

\begin{table}[t]
  \caption{Summary of notation used in the orbital-rotation-shadow overlap construction.}
  \label{tab:core_notation}
  \centering
  \footnotesize
  \setlength{\tabcolsep}{5pt}
  \renewcommand{\arraystretch}{1.08}
  \begin{tabular}{@{}p{0.26\columnwidth}p{0.66\columnwidth}@{}}
    \toprule
    \textbf{Symbol} & \textbf{Description} \\
    \midrule
    $N,\eta,N_{\mathrm{aug}}$
      & Numbers of physical modes, particles, and augmented modes,
        respectively, with $N_{\mathrm{aug}}=N+\eta$. \\
    $[m]$
      & Index set $\{1,\ldots,m\}$. \\
    $\SN{m}{\eta}$
      & Set of $\eta$-element occupation patterns on $m$ modes. \\
    $\ket{\Psi_T},\ket{\phi(\Phi)}$
      & Physical $\eta$-particle quantum state and target Slater determinant
        on the $N$ physical modes. \\
    $\Phi$
      & $N\times\eta$ occupied-orbital matrix representing
        $\ket{\phi(\Phi)}$ or an AFQMC walker. \\
    $\vec a,\ket{\Psi_{\mathrm{aug}}}$
      & Ancillary occupation pattern $\vec a=N+[\eta]$ and coherent augmented
        state $\ket{\Psi_{\mathrm{aug}}}=(\ket{\Psi_T}+\ket{\vec a})/\sqrt{2}$. \\
    $u,U_\eta(u)$
      & Haar-random $N_{\mathrm{aug}}\times N_{\mathrm{aug}}$ orbital rotation and its action on the
        fixed-$\eta$ sector. \\
    $\vec z$
      & Measured occupation pattern in $\SN{N_{\mathrm{aug}}}{\eta}$; one shadow record is
        $(u,\vec z)$. \\
    $U_o,U_a$
      & Measured-row blocks $U_o=u_{\vec z,[N]}\in\mathbb{C}^{\eta\times N}$
        and $U_a=u_{\vec z,\vec a}\in\mathbb{C}^{\eta\times\eta}$. \\
    $S,H$
      & Shadow-defined matrices $S=U_o^\dagger U_a\in
        \mathbb{C}^{N\times\eta}$ and $H=S^\dagger\Phi\in
        \mathbb{C}^{\eta\times\eta}$. \\
    $o_\phi,D(\Phi)$
      & Conjugate overlap conventions
        $o_\phi=\braket{\phi}{\Psi_T}$ and
        $D(\Phi)=\braket{\Psi_T}{\phi(\Phi)}$. \\
    \bottomrule
  \end{tabular}

  \vspace{3pt}
  \begin{minipage}{0.94\columnwidth}
    \scriptsize
    \raggedright
    We append $\eta$ ancillary modes, giving $N_{\mathrm{aug}}=N+\eta$.
    The subscript $o$ in $U_o$ labels the original physical-mode block,
    not occupied orbitals. Before augmentation, $u\in\mathcal U_N$ and
    $\vec z\in\SN{N}{\eta}$. In the augmented overlap protocol,
    $u\in\mathcal U_{N_{\mathrm{aug}}}$ and
    $\vec z\in\SN{N_{\mathrm{aug}}}{\eta}$.
  \end{minipage}
\end{table}

\subsection{Orbital-Rotation Shadow Records}

For an $\eta$-particle state on $N$ modes, let orbital occupation patterns
$\vec p,\vec r\in\mathcal S_{N,\eta}$. A one-particle rotation
$u\in\mathcal U_N$ induces
\begin{equation}
    U_\eta(u)
    =
    \bigwedge^\eta u,
    \qquad
    \bra{\vec p}U_\eta(u)\ket{\vec r}
    =
    \det\!\left(u_{\vec p,\vec r}\right).
\end{equation}

An orbital-rotation shadow is obtained by sampling $u$ from the Haar measure
on $\mathcal U_N$, applying $U_\eta(u)$, and measuring the occupied modes.
The outcome is an ordered occupation pattern
\begin{equation}
    \vec z
    =
    (z_1,\ldots,z_\eta)
    \in
    \mathcal S_{N,\eta},
    \quad
    1
    \leq
    z_1
    <
    \cdots
    <
    z_\eta
    \leq
    N.
\end{equation}
One shadow record is a pair $(u,\vec z)$.

\subsection{Augmented Overlap Observable}
\label{sec:augmented_observable}

In order to compute the overlaps between the trial state and a Slater determinant, one could first turn this overlap to a high-order reduced density matrix (RDM) of an augmented state, and use particle-number-conserving shadows to evaluate the RDM~\cite{low24_08964}.

We append
$\eta$ ancillary modes and introduce a known $\eta$-particle Slater determinant
supported entirely on those modes. The physical state and the reference Slater
determinant then occupy disjoint mode blocks within the same fixed-$\eta$
particle sector. Their coherent superposition converts the desired complex
overlap into an off-diagonal coherence. The random orbital rotation acts on
the full physical-plus-ancillary space, allowing this relative phase to
appear in the occupation statistics.

We derive the estimator by specializing the fixed-particle inverse-channel
construction of Ref.~\cite{low24_08964} to the augmented coherence. After
expressing the $k=\eta$ inverse observable through a generating polynomial, we
evaluate the fixed-particle trace and show that only its degree-$\eta$
coefficient survives. The resulting one-shot estimator reduces to a single
$\eta\times\eta$ matrix determinant.

Consider an $N$-mode fermionic state $\ket{\Psi_T}$ with a fixed particle number
$\eta$. The goal is to estimate the overlap

\begin{equation}
    \label{eq:target_overlap}
    o_\phi
    =
    \braket{\phi}{\Psi_T}
\end{equation}
for a Slater determinant
\begin{equation}
    \label{eq:target_slater}
    \ket{\phi}
    =
    U_\eta(w)\ket{\vec q},
    \qquad
    \vec q
    \in
    \mathcal S_{N,\eta},
\end{equation}
where $w\in\mathcal U_N$ is a one-particle orbital rotation. Let $w_j$ denote the $j$-th column of $w$. The occupied-orbital matrix is
\begin{equation}
    \Phi
    =
    \begin{bmatrix}
        w_{q_1} & \cdots & w_{q_\eta}
    \end{bmatrix}
    =
    w_{:,\vec q}
    \in
    \mathbb C^{N\times\eta}.
\end{equation}

We now implement this augmentation. Appending $\eta$ ancillary modes gives
\begin{equation}
    \label{eq:augmented_modes}
    N_{\mathrm{aug}}
    =
    N+\eta.
\end{equation}
The ancillary occupation pattern is
\begin{equation}
    \label{eq:ancilla_modes}
    \vec a
    =
    N+[\eta]
    =
    (N+1,\ldots,N+\eta).
\end{equation}

Extend the target rotation to the augmented one-particle space as
\begin{equation}
    \label{eq:augmented_rotation}
    W
    =
    w\oplus I_\eta.
\end{equation}
We regard all three states as elements of the same $N_{\mathrm{aug}}$-mode, fixed-$\eta$
particle sector.  
We
use $\ket{\vec a}$ for the Slater determinant with every physical mode unoccupied
and all $\eta$ ancillary modes occupied. The physical and ancillary branches
therefore have support on disjoint mode blocks.

The augmented state is
\begin{equation}
    \label{eq:augmented_state}
    \ket{\Psi_{\mathrm{aug}}}
    =
    \frac{1}{\sqrt{2}}
    \left(
        \ket{\Psi_T}\otimes\ket{0}^{\otimes\eta}
        +
        \ket{\vec a}
    \right).
\end{equation}

Note that the appended modes provide a coherent phase reference; they are not additional basis
functions for the physical Hamiltonian.

In the shadow measurement below, the Haar-random rotation
$u\in\mathcal U_{N_{\mathrm{aug}}}$ acts on this full physical-plus-ancillary space. It is
distinct from the fixed target rotation $W$.

Define a number-conserving transition operator
\begin{equation}
    \label{eq:D_operator}
    D_{\vec q}^{\vec a}
    =
    a_{a_1}^{\dagger}
    \cdots
    a_{a_\eta}^{\dagger}
    a_{q_\eta}
    \cdots
    a_{q_1}.
\end{equation}
With occupation Slater determinants ordered by increasing mode index,
$D_{\vec q}^{\vec a}$ acts on the fixed-$\eta$ sector as
\begin{equation}
    D_{\vec q}^{\vec a}
    =
    \ketbra{\vec a}{\vec q}.
\end{equation}

Because $W$ rotates the physical Slater determinant while leaving the ancillary
Slater determinant unchanged,
\begin{equation}
    \label{eq:rotated_coherence}
    C_\phi
    =
    U_\eta(W)
    D_{\vec q}^{\vec a}
    U_\eta^\dagger(W)
    =
    \ket{\vec a}\!\bra{\phi}.
\end{equation}

The physical and ancillary branches have disjoint occupation support, so
\begin{equation}
    \braket{\phi}{\vec a}
    =
    \braket{\Psi_T}{\vec a}
    =
    0.
\end{equation}
Consequently,
\begin{equation}
    \label{eq:overlap_identity}
    \begin{aligned}
        \tr[\wedge^\eta]{
            C_\phi
            \ketbra{\Psi_{\mathrm{aug}}}{\Psi_{\mathrm{aug}}}
        }
        &=
        \braket{\phi}{\Psi_{\mathrm{aug}}}
        \braket{\Psi_{\mathrm{aug}}}{\vec a}
        \\
        &=
        \frac{1}{2}
        \braket{\phi}{\Psi_T}.
    \end{aligned}
\end{equation}
Here $\tr[\wedge^\eta]{\cdot}$ denotes the trace over the fixed-$\eta$
particle sector.

The resulting estimator can be stated directly in terms of measured
records. For each record, prepare $\ket{\Psi_{\mathrm{aug}}}$, sample a
Haar-random orbital rotation $u\in\mathcal U_{N_{\mathrm{aug}}}$, apply $U_\eta(u)$, and
measure the occupation pattern $\vec z$. Given the occupied-orbital matrix
$\Phi\in\mathbb C^{N\times\eta}$ representing the target Slater determinant,
retain the measured-row blocks
$U_o=u_{\vec z,[N]}\in\mathbb C^{\eta\times N}$ and
$U_a=u_{\vec z,\vec a}\in\mathbb C^{\eta\times\eta}$.
The record contributes the complex number
\[
\widehat O_{\vec q,w}(u,\vec z)
=
2\binom{N_{\mathrm{aug}}+1}{\eta}
\det\!\left[(U_o\Phi)^\dagger U_a\right].
\]
Averaging independent records gives the overlap estimate
\[
\widehat o_\phi
=
\frac{1}{N_{\mathrm{shadow}}}
\sum_{t=1}^{N_{\mathrm{shadow}}}
\widehat O_{\vec q,w}(u^{(t)},\vec z^{(t)}).
\]
In the following sections we will provide a detailed derivation of this result. 

\subsection{Orbital-Rotation Shadow Estimator}
\label{sec:shadow_estimator}

Suppose we have already performed orbital-rotation shadow measurements and
stored the records $(u,\vec z)$. We now use these records to estimate the
augmented coherence expectation in \cref{eq:overlap_identity}. The fixed
target rotation $W$ is absorbed into the effective rotation $g_{\vec z}$
defined below.

Let $v_{\vec z}$ be any permutation matrix that maps the first $\eta$ modes
to the measured occupation pattern $\vec z$, so that
\begin{equation}
    \left(v_{\vec z}^{\dagger}u\right)_{[\eta],:}
    =
    u_{\vec z,:}.
\end{equation}
Its action on the remaining modes may be chosen arbitrarily. We use
\begin{equation}
    \Pi_{[\eta]}
    =
    \ketbra{[\eta]}{[\eta]}
\end{equation}
for the many-particle projector onto the Slater determinant occupying the first
$\eta$ modes. We then define
\begin{equation}
    \label{eq:g_definition}
    g_{\vec z}
    =
    v_{\vec z}^{\dagger}uW.
\end{equation}

Let $E_{\eta,\eta}$ denote the $k=\eta$ inverse-channel observable of
Ref.~\cite[Eq.~(84)]{low24_08964}. The fixed target rotation $W$ can be
absorbed into the inverse-channel rotation. Using
\cref{eq:rotated_coherence}, the representation property of $U_\eta$, and
cyclicity of the trace gives
\begin{equation}
    \label{eq:absorb_target_rotation}
    \begin{aligned}
        &\tr[\wedge^\eta]{
            C_\phi
            U_\eta^\dagger(v_{\vec z}^\dagger u)
            E_{\eta,\eta}
            U_\eta(v_{\vec z}^\dagger u)
        }
        \\
        &=
        \tr[\wedge^\eta]{
            D_{\vec q}^{\vec a}
            U_\eta^\dagger(g_{\vec z})
            E_{\eta,\eta}
            U_\eta(g_{\vec z})
        }.
    \end{aligned}
\end{equation}

The resulting single-shot estimator is
\begin{align}
    \label{eq:single_shot_shadow_estimator}
    \widehat{\Omega}_{\vec q,w}(u,\vec z)
    &=
    \tr[\wedge^\eta]{
        D_{\vec q}^{\vec a}
        U_\eta^\dagger(g_{\vec z})
        E_{\eta,\eta}
        U_\eta(g_{\vec z})
    }
    \\
    &=
    \sum_{s=0}^{\eta}    E'_{\eta,\eta,s}
    \nonumber\\
    &\qquad{}\times
    \tr[\wedge^\eta]{
        D_{\vec q}^{\vec a}
        U_\eta^\dagger(g_{\vec z})
        \operatorname{Sim}_{\eta,s}(\Pi_{[\eta]})
        U_\eta(g_{\vec z})
    },
    \nonumber
\end{align}
in which we reused the similarity projector introduced in Ref.~\cite[Eq.~(85)]{low24_08964} as
\begin{equation}
\label{eq:similarity_projector}
\operatorname{Sim}_{\eta,s}(\Pi_{[\eta]})
=
\sum_{\vec r\in\SN{N_{\mathrm{aug}}}{\eta}:|\vec r\cap[\eta]|=s}
\ketbra{\vec r}{\vec r},
\end{equation}
which sums over $\vec r$ that share the same value $s=|\vec r\cap[\eta]|$. 

The expectation is over the Haar-random rotation $u$ and the
Born-distributed outcome $\vec z$ obtained by measuring
$U_\eta(u)\ket{\Psi_{\mathrm{aug}}}$. Inverse-channel unbiasedness and
\cref{eq:overlap_identity} give expectation value
\begin{equation}
    \label{eq:augmented_estimator_unbiased}
    \mathbb E_{u,\vec z}
    \left[
        \widehat{\Omega}_{\vec q,w}(u,\vec z)
    \right]
    =
    \tr[\wedge^\eta]{
        C_\phi
        \ketbra{\Psi_{\mathrm{aug}}}{\Psi_{\mathrm{aug}}}
    }
    =
    \frac{1}{2}o_\phi.
\end{equation}

Using $s=|\vec r\cap[\eta]|$, the coefficient is the $k=\eta$
specialization of Ref.~\cite[Eq.~(84)]{low24_08964},
\begin{equation}
\label{eq:eprime_eta}
E'_{\eta,\eta,s}
=E_{\eta,\eta,\eta-s}
=(-1)^{\eta-s}
\binom{\eta}{s}^{-1}
\binom{N_{\mathrm{aug}}-\eta+s}{s}.
\end{equation}

We now introduce the Majorana operators and write creation and annihilation operators with them, 
\begin{equation}
\label{eq:majorana_definition}
\begin{aligned}
a_p&=\frac{\gamma_{2p-1}+i\gamma_{2p}}{2},
&
a_p^{\dagger}&=\frac{\gamma_{2p-1}-i\gamma_{2p}}{2},
\\
\beta_p&=\gamma_{2p-1}\gamma_{2p}.
\end{aligned}
\end{equation}
Then $\hat n_p=(1+i\beta_p)/2$. Operators $\beta_p$ commute for
distinct modes, so we define their $j$-th elementary symmetric polynomial as
\begin{equation}
    e_j(\{\beta\}_{[\eta]})
    =
    \sum_{1\leq p_1<\cdots<p_j\leq\eta}
    \beta_{p_1}\cdots\beta_{p_j},
    \qquad
    e_0(\{\beta\}_{[\eta]})
    =
    1.
\end{equation}
With this notation, Ref.~\cite[Eq.~(95)]{low24_08964} gives
\begin{align}
\label{eq:sim_beta_expansion}
\operatorname{Sim}_{\eta,s}(\Pi_{[\eta]})
&=
(-1)^s
\sum_{j=0}^{\eta}
f_{\eta,s}(j)i^j e_j(\{\beta\}_{[\eta]}),
\\
\label{eq:f_coefficients}
f_{\eta,s}(j)
&=
\sum_{x=j}^{\eta}
(-1)^x\binom{x}{s}2^{-x}
\binom{\eta-j}{\eta-x}.
\end{align}
Now we define
\begin{align}
\label{eq:generating_polynomial}
P(\kappa, {\vec z})
&=
\tr[\wedge^{\eta}]{
D_{\vec q}^{\vec a}
U_{\eta}^{\dagger}(g_{\vec z})
\rho_{\eta}(\kappa)
U_{\eta}(g_{\vec z})},\\
\label{eq:rho_eta}
\rho_{\eta}(\kappa)
&=
\prod_{\ell=1}^{\eta}(1+i\kappa\beta_{\ell})
=
\sum_{j=0}^{\eta}(i\kappa)^j e_j(\{\beta\}_{[\eta]}).
\end{align}
Expanding $P(\kappa,\vec z)$ in powers of $\kappa$ defines
\begin{equation}
    \label{eq:pg_coefficients}
    \begin{aligned}
        P(\kappa,\vec z)
        &=
        \sum_{j=0}^{\eta}
        b_j(g_{\vec z})\kappa^j,
        \\
        b_j(g_{\vec z})
        &=
        \tr[\wedge^\eta]{
            D_{\vec q}^{\vec a}
            U_\eta^\dagger(g_{\vec z})
            i^j e_j(\{\beta\}_{[\eta]})
            U_\eta(g_{\vec z})
        }.
    \end{aligned}
\end{equation}
Substituting these coefficients into
\cref{eq:single_shot_shadow_estimator,eq:sim_beta_expansion} gives
\begin{equation}
\label{eq:estimator_from_coefficients}
\widehat{\Omega}_{\vec q,w}(u,\vec z)
=
\sum_{s=0}^{\eta}
E'_{\eta,\eta,s}(-1)^s
\sum_{j=0}^{\eta}f_{\eta,s}(j)b_j(g_{\vec z}).
\end{equation}
The desired single-shot overlap estimator is
\begin{equation}
\label{eq:overlap_single_shot_general}
\widehat O_{\vec q,w}(u,\vec z)
=2\widehat{\Omega}_{\vec q,w}(u,\vec z).
\end{equation}

\subsection{Fixed-Particle Trace and Full-Fock Trace}
\label{sec:fixed_full_trace}

We now discuss how to evaluate \cref{eq:generating_polynomial} efficiently. All operators in the trace are either fermionic Gaussian operators or density matrices, or products of linear combinations of Majorana operators. Such traces can be evaluated efficiently using Grassmann integration~\cite{wan2023matchgate}.
However, Grassmann integration naturally evaluates traces over the full Fock space, whereas the trace in \cref{eq:generating_polynomial} is restricted to the fixed-particle sector. The following theorem relates the full-space trace to the fixed-particle trace.

\begin{lemma}
\label{thm:rho_is_gamma}
For a one-particle matrix
$A\in\mathbb C^{N_{\mathrm{aug}}\times N_{\mathrm{aug}}}$,
define its number-conserving second quantization
$\Gamma(A)$ by
\begin{equation}
    \label{eq:Gamma_exterior_power}
    \langle\vec p|\Gamma(A)|\vec q\rangle
    :=
    \begin{cases}
        \det A_{\vec p,\vec q},
        & |\vec p|=|\vec q|,\\
        0,
        & |\vec p|\neq|\vec q|.
    \end{cases}
\end{equation}
Here $\vec p$ and $\vec q$ label occupation vectors,
and $A_{\vec p,\vec q}$ is the submatrix whose row
and column indices are the occupied modes in
$\vec p$ and $\vec q$, respectively, ordered increasingly. If $A$ is invertible, we can write the operator as
\begin{equation}
    \label{eq:Gamma_d_exp}
    \Gamma(A)
    =
    \exp\!\left(
        \sum_{p,q=1}^{N_{\mathrm{aug}}}
        [\log A]_{pq}\,a_p^\dagger a_q
    \right).
\end{equation}
Now let
\begin{equation}
    \label{eq:P_projector}
    P_{[\eta]}
    =
    \sum_{\ell=1}^{\eta}
    \ketbra{\ell}{\ell}
\end{equation}
be the one-particle projector onto the first $\eta$ modes. For
$\kappa\neq1$, set
$\alpha=(1-\kappa)^\eta$ and
$\mu=2\kappa/(1-\kappa)$, and define
\begin{equation}
    d
    =
    I+\mu P_{[\eta]}
    =
    \begin{pmatrix}
        (1+\mu)I_\eta & 0 \\
        0 & I_{N_{\mathrm{aug}}-\eta}
    \end{pmatrix}.
\end{equation}
With the above definitions we show that,
\begin{equation}
    \rho_\eta(\kappa)
    =
    \prod_{\ell=1}^{\eta}
    \left(1+i\kappa\beta_\ell\right)
    =
    \alpha\Gamma(d).
\end{equation}

\end{lemma}

\begin{proof}

Since
\begin{equation}
    i\beta_\ell
    =
    2\hat n_\ell-1,
\end{equation}
each factor satisfies
\begin{equation}
    1+i\kappa\beta_\ell
    =
    (1-\kappa)
    \left(
        1+\mu\hat n_\ell
    \right).
\end{equation}
Therefore,
\begin{equation}
    \rho_\eta(\kappa)
    =
    (1-\kappa)^\eta
    \prod_{\ell=1}^{\eta}
    \left(
        1+\mu\hat n_\ell
    \right).
\end{equation}
Because $d=I+\mu P_{[\eta]}$ multiplies the first $\eta$ one-particle
basis vectors by $1+\mu$ and leaves the remaining basis vectors unchanged,
\begin{equation}
    \Gamma(d)
    =
    \prod_{\ell=1}^{\eta}
    \left(
        1+\mu\hat n_\ell
    \right).
\end{equation}
Combining these identities gives
\begin{equation}
    \label{eq:identity_rho_gamma}
    \rho_\eta(\kappa)
    =
    \alpha\Gamma(d).
\end{equation}

\end{proof}

\begin{theorem}[Fixed-particle trace versus full-Fock trace]
\label{thm:fixed_full_trace_conversion}
Let
\begin{equation}
\label{eq:fixed_trace_T}
T_{\eta}(g,\kappa)
=
\tr[\wedge^\eta]{
D_{\vec q}^{\vec a}U_{\eta}^{\dagger}(g)
\rho_{\eta}(\kappa)U_{\eta}(g)}
\end{equation}
and let $T_{\Fock}(g,\kappa)$ denote the same operator traced over the
full Fock space $\Fock_{N_{\mathrm{aug}}}$, with $U(g)$ acting on all particle sectors.
Then
\begin{equation}
\label{eq:full_fixed_conversion}
T_{\Fock}(g,\kappa)=2^{N_{\mathrm{aug}}-2\eta}T_{\eta}(g,\kappa).
\end{equation}
\end{theorem}

\begin{proof}

We define the following by absorbing the rotation
\begin{equation}
\label{eq:h_and_K}
\begin{aligned}
h&=g^{\dagger}(I+\mu P_{[\eta]})g=I+\mu K,
&
K&=g^{\dagger}P_{[\eta]}g
\end{aligned}
\end{equation}
Note that $K^2=K$ and $\rank K=\eta$.

On the fixed $\eta$-particle sector,
\begin{equation}
\label{eq:fixed_trace_minor}
\begin{aligned}
T_{\eta}(g,\kappa)&=\alpha\tr[\wedge^\eta]{
D_{\vec q}^{\vec a}\Gamma(h)} \\
&=\alpha \tr[\wedge^\eta]{
\Gamma(h)\ket{\vec{a}}\bra{\vec{q}}} \\
&=\alpha \tr[\wedge^\eta]{
\sum_{\vec{p}}{\det(h_{\vec p,\vec a})}\ket{\vec{p}}\bra{\vec{q}}} \\
&=\alpha\det(h_{\vec q,\vec a}). \\
\end{aligned}
\end{equation}
where we used \cref{thm:rho_is_gamma} and the cyclicity of the trace, and to remind the readers, $\vec q$ and $\vec a$ are defined in Eq.~\ref{eq:D_operator}. Note that $\kappa$ no longer appears explicitly in the RHS as we have applied the identity in Eq.~\ref{eq:identity_rho_gamma}. 

For the full Fock trace, the standard number-conserving Gaussian trace
identities give~\cite{Koonin97_1}
\begin{equation}
    \label{eq:gaussian_trace}
    \tr[\Fock]{\Gamma(h)}
    =
    \det(I+h),
\end{equation}
and
\begin{equation}
    \frac{
        \tr[\Fock]{a_i^\dagger a_j\Gamma(h)}
    }{
        \tr[\Fock]{\Gamma(h)}
    }
    =
    (C_h)_{j,i},
\end{equation}
where the Gaussian contraction matrix is
\begin{equation}
    \label{eq:gaussian_contraction}
    C_h
    =
    h(I+h)^{-1}.
\end{equation}

Wick's theorem applied to $D_{\vec q}^{\vec a}$ gives
\begin{equation}
    \label{eq:full_trace_minor}
    T_{\Fock}(g,\kappa)
    =
    \alpha
    \det(I+h)
    \det\!\left((C_h)_{\vec q,\vec a}\right).
\end{equation}
Since $K$ is a rank-$\eta$ projector,
\begin{equation}
    \label{eq:det_I_h}
    \det(I+h)
    =
    \det(2I+\mu K)
    =
    2^{N_{\mathrm{aug}}-\eta}(2+\mu)^\eta.
\end{equation}
Also,
\begin{equation}
    \label{eq:gaussian_contraction_projector}
    C_h
    =
    (I+\mu K)(2I+\mu K)^{-1}
    =
    \frac{1}{2}(I-K)
    +
    \frac{1+\mu}{2+\mu}K.
\end{equation}
The sets $\vec q$ and $\vec a$ are disjoint, so the identity block does not
contribute to $(C_h)_{\vec q,\vec a}$. Thus
\begin{equation}
    \label{eq:gaussian_contraction_block}
    (C_h)_{\vec q,\vec a}
    =
    \frac{1}{2(2+\mu)}
    h_{\vec q,\vec a}.
\end{equation}
Substituting
\cref{eq:gaussian_contraction_block,eq:det_I_h}
into \cref{eq:full_trace_minor} yields
\begin{equation}
\label{eq:full_trace_conversion_final}
T_{\Fock}(g,\kappa)
=
2^{N_{\mathrm{aug}}-2\eta}\alpha\det(h_{\vec q,\vec a})
=
2^{N_{\mathrm{aug}}-2\eta}T_{\eta}(g,\kappa).
\end{equation}
Both sides are polynomials in $\kappa$, so the identity extends to
$\kappa=1$ by polynomial continuation.
\end{proof}

\subsection{Direct Matrix Determinant Estimator}
\label{sec:direct_determinant_estimator}

The preceding theorem relates the fixed-particle and full-Fock traces. The
matrix determinant reduction is most transparent in the fixed-particle
representation. Comparing
\cref{eq:generating_polynomial,eq:fixed_trace_T} gives
\begin{equation}
    \label{eq:P_as_fixed_trace}
    P(\kappa,\vec z)
    =
    T_\eta(g_{\vec z},\kappa).
\end{equation}

The Gaussian representation from \cref{thm:rho_is_gamma} is
\begin{equation}
    \label{eq:rho_alpha_gamma_again}
    \begin{aligned}
        \rho_\eta(\kappa)
        &=
        \alpha\Gamma(d),
        &
        \alpha
        &=
        (1-\kappa)^\eta,
        \\
        d
        &=
        I+\mu P_{[\eta]},
        &
        \mu
        &=
        \frac{2\kappa}{1-\kappa}.
    \end{aligned}
\end{equation}
Substituting $g=g_{\vec z}$ into \cref{eq:fixed_trace_minor} gives
\begin{equation}\label{eq:Pg_as_minor}
\begin{aligned}
P(\kappa, \vec z)
&=
(1-\kappa)^{\eta}
\det\!\left((I+\mu K^{(\vec z)})_{\vec q,\vec a}\right),
\\
K^{(\vec z)} &= g_{\vec z}^{\dagger}P_{[\eta]}g_{\vec z}.
\end{aligned}
\end{equation}

Since $\vec q$ and $\vec a$ are disjoint,
\begin{equation}
\label{eq:offdiag_identity_vanishes}
(I+\mu K^{(\vec z)})_{\vec q,\vec a}=\mu K^{(\vec z)}_{\vec q,\vec a}.
\end{equation}
Therefore
\begin{equation}
\label{eq:Pg_monomial}
P(\kappa, \vec z)
=(2\kappa)^{\eta}\det(K^{(\vec z)}_{\vec q,\vec a}),
\end{equation}
and all lower-order coefficients vanish:
\begin{equation}
\label{eq:leading_coefficient}
b_j(g_{\vec z})=0\quad(j<\eta),
\qquad
b_{\eta}(g_{\vec z})=2^{\eta}\det(K^{(\vec z)}_{\vec q,\vec a}).
\end{equation}
Substitution into \cref{eq:estimator_from_coefficients}, together with
\begin{equation}
\label{eq:f_eta_s_eta}
f_{\eta,s}(\eta)=(-1)^{\eta}\binom{\eta}{s}2^{-\eta},
\end{equation}
gives
\begin{align}
\label{eq:direct_augmented_estimator}
\widehat{\Omega}_{\vec q,w}(u,\vec z)
&=
2^{-\eta}b_{\eta}(g_{\vec z})
\sum_{s=0}^{\eta}\binom{N_{\mathrm{aug}}-\eta+s}{s}
\\
&=
\binom{N_{\mathrm{aug}}+1}{\eta}\det(K^{(\vec z)}_{\vec q,\vec a}).
\nonumber
\end{align}
Thus the desired overlap has the single-shot estimator
\begin{equation}
\label{eq:direct_overlap_single_shot}
\widehat O_{\vec q,w}(u,\vec z)
=
2\binom{N_{\mathrm{aug}}+1}{\eta}\det(K^{(\vec z)}_{\vec q,\vec a}).
\end{equation}

\subsection{Computing the Matrix Determinant Block}
\label{sec:computing_K}

Now we provide an efficient way to evaluate Eq.~\ref{eq:direct_overlap_single_shot} without forming the full $K$ matrix. We first note that only the $\eta\times\eta$ block $K_{\vec q,\vec a}$ is needed. Since $K=g^\dagger P_{[\eta]}g$ and $P_{[\eta]}$ projects onto the first
$\eta$ one-particle modes, let
\begin{equation}
\label{eq:G_rows}
G=g_{[\eta],:}
\end{equation}
and
\begin{equation}
\label{eq:K_from_G}
K=G^{\dagger}G.
\end{equation}
The permutation $v_{\vec z}$ maps $[\eta]$ to $\vec z$, so
\begin{equation}
\label{eq:G_measured_rows}
G=(uW)_{\vec z,:}
\end{equation}
and, as a result, 
\begin{equation}
\label{eq:K_block_uW}
K_{\vec q,\vec a}
=
\left((uW)_{\vec z,\vec q}\right)^{\dagger}
(uW)_{\vec z,\vec a}.
\end{equation}
Because $\vec q$ lies in the original system and $\vec a$ lies in the
ancilla modes, we can avoid forming $W$ explicitly and compute $K_{\vec q,\vec a}$ as
\begin{equation}
\label{eq:K_block_direct}
K_{\vec q,\vec a}
=
\left(u_{\vec z,[N]}w_{:,\vec q}\right)^{\dagger}
u_{\vec z,\vec a}.
\end{equation}
If the target Slater determinant is represented as the $N\times\eta$ orbital
coefficient matrix
\begin{equation}
\label{eq:Phi_target}
\Phi=w_{:,\vec q},
\end{equation}
then
\begin{equation}
\label{eq:K_block_Phi}
K_{\vec q,\vec a}
=
\left(u_{\vec z,[N]}\Phi\right)^{\dagger}
u_{\vec z,\vec a}.
\end{equation}
The derivation above used an orbital matrix with orthonormal columns,
$\Phi=w_{:,\vec q}$. The same estimator applies to a general AFQMC walker.
For a full-column-rank matrix, let
\begin{equation}
    \Phi
    =
    QR,
    \qquad
    Q^\dagger Q
    =
    I_\eta
\end{equation}
be a thin QR factorization. Since $Q$ has orthonormal columns, it can be completed to a unitary one-particle rotation and therefore satisfies the assumptions of the
derivation above. The corresponding Slater determinant satisfies
\begin{equation}
    \ket{\phi(\Phi)}
    =
    \det(R)\ket{\phi(Q)}.
\end{equation}
The single-shadow matrix determinant acquires the same factor in the bra orientation:
\begin{equation}
    \det\!\left[
        \left(u_{\vec z,[N]}\Phi\right)^\dagger
        u_{\vec z,\vec a}
    \right]
    =
    \det(R)^*
    \det\!\left[
        \left(u_{\vec z,[N]}Q\right)^\dagger
        u_{\vec z,\vec a}
    \right].
\end{equation}
Thus the estimator remains valid for an arbitrary full-column-rank walker, interpreted as the ordered, generally unnormalized Slater determinant $\ket{\phi(\Phi)}$. If $\Phi$ is rank deficient, both the Slater determinant state and the single-shadow matrix determinant vanish.

\subsection{Overlap Algorithm}
\label{sec:overlap_algorithm}
Algorithm~1 summarizes the orbital-rotation-shadow protocol for estimating overlaps with Slater determinants.
\begin{widetext}
\noindent\rule{\textwidth}{0.4pt}
\noindent\textbf{Algorithm 1: Estimating Slater-determinant overlaps with orbital-rotation shadows}

\smallskip
\noindent\textbf{Input:} an $N$-mode, $\eta$-particle state $\ket{\Psi_T}$;
target Slater determinants
$\ket{\phi_i}$ represented by orbital matrices
$\Phi_i\in\mathbb{C}^{N\times\eta}$ for $i=1,\ldots,L$; number of
shadow samples $N_{\mathrm{shadow}}$.

\smallskip
\noindent\textbf{Output:} estimates $\widehat o_i$ of
$o_i=\braket{\phi_i}{\Psi_T}$. 

\begin{enumerate}[leftmargin=1.2cm,itemsep=2pt]
\item Set $N_{\mathrm{aug}}=N+\eta$ and $\vec a=N+[\eta]$.
\item For $t=1,\ldots,N_{\mathrm{shadow}}$:
\begin{enumerate}[label=\alph*.,leftmargin=1.1cm,itemsep=1pt]
\item prepare
\begin{equation}
\label{eq:algorithm_augmented_state}
\ket{\Psi_{\mathrm{aug}}}=(\ket{\Psi_T}\otimes\ket{0}^{\otimes\eta}+\ket{\vec a})/\sqrt{2}.
\end{equation}
\item sample $u^{(t)}$ from Haar measure on $\mathcal{U}_{N_{\mathrm{aug}}}$;
\item apply $U_{\eta}(u^{(t)})$ to $\ket{\Psi_{\mathrm{aug}}}$;
\item measure in the occupation basis to obtain
$\vec z^{(t)}\in\SN{N_{\mathrm{aug}}}{\eta}$;
\item store $(u^{(t)},\vec z^{(t)})$.
\end{enumerate}
\item For each target $i$ and shadow $t$, form
\begin{equation}
\label{eq:algorithm_A_B_blocks}
A_i^{(t)}=u^{(t)}_{\vec z^{(t)},[N]}\Phi_i,
\qquad
B^{(t)}=u^{(t)}_{\vec z^{(t)},\vec a}.
\end{equation}
\item Compute
\begin{equation}
\label{eq:algorithm_single_shot}
\widehat O_i^{(t)}
=
2\binom{N_{\mathrm{aug}}+1}{\eta}
\det\!\left((A_i^{(t)})^{\dagger}B^{(t)}\right).
\end{equation}
\item Return the shadow mean
\begin{equation}
\label{eq:algorithm_shadow_mean}
\widehat o_i=\frac{1}{N_{\mathrm{shadow}}}\sum_{t=1}^{N_{\mathrm{shadow}}}\widehat O_i^{(t)}.
\end{equation}
\end{enumerate}
\noindent\rule{\textwidth}{0.4pt}
\end{widetext}

\subsection{Overlap Derivatives}
\label{sec:overlap_derivatives}

Algorithm~1 estimates the bra-oriented overlap
$o_\phi=\braket{\phi(\Phi)}{\Psi_T}$. QC-AFQMC uses the conjugate orientation
$D(\Phi)=\braket{\Psi_T}{\phi(\Phi)}=o_\phi^*$, which we adopt below.
Force biases and local energies require ratios of derivatives of $D(\Phi)$. The force bias and the local energy share the same structure as 
\begin{equation}
    \label{eq:force_bias_loc_eng}
    \frac{\bra{\Psi_T}
    \widehat O
    \ket{\phi(\Phi)}}{\braket{\Psi_T}{\phi(\Phi)}}
\end{equation}
For the force bias, $\widehat O$ is a one-body operator.
For the local energy, it is the Hamiltonian, evaluated using
one-body operators and their products.
Let a QC-AFQMC walker be represented by
\begin{equation}
\label{eq:walker_Phi}
\Phi\in\mathbb{C}^{N\times\eta},
\qquad
\ket{\phi(\Phi)}
=
\prod_{\alpha=1}^{\eta}
\left(\sum_{p=1}^{N}\Phi_{p\alpha}a_p^{\dagger}\right)\ket{0}.
\end{equation}
For a single shadow, we define
\begin{equation}
\label{eq:S_block}
U_o=u_{\vec z,[N]},
\qquad
U_a=u_{\vec z,\vec a},
\qquad
S=U_o^{\dagger}U_a.
\end{equation}
For a single shadow, the estimator of the denominator in Eq.~\ref{eq:force_bias_loc_eng} is
\begin{equation}
\label{eq:D_single_shadow}
\widehat D_{u,\vec z}(\Phi)
=
2\binom{N_{\mathrm{aug}}+1}{\eta}\det H,
\qquad
H=S^{\dagger}\Phi.
\end{equation}

Let $X\in\mathbb{C}^{N\times N}$ be a one-particle generator. Its
number-conserving Fock-space representation is
\begin{equation}
    \label{eq:X_generator}
    \widehat X
    =
    \sum_{p,q=1}^{N}
    X_{pq}a_p^{\dagger}a_q.
\end{equation}
The corresponding transformation of the walker is
\begin{equation}
    e^{\lambda\widehat X}\ket{\phi(\Phi)}
    =
    \ket{\phi\!\left(e^{\lambda X}\Phi\right)}.
\end{equation}
Define
\begin{equation}
    \Phi(\lambda)
    =
    e^{\lambda X}\Phi.
\end{equation}
in which $\lambda$ is an auxiliary scalar variable. 
Differentiating the exact overlap gives
\begin{equation}
    \label{eq:one_body_numerator}
    \left.
    \partial_\lambda
    D(\Phi(\lambda))
    \right|_{\lambda=0}
    =
    \bra{\Psi_T}
    \widehat X
    \ket{\phi(\Phi)}.
\end{equation}
Thus the overlap derivative is precisely the one-body matrix-element
numerator required in QC-AFQMC. The shadow records depend only on the
augmented trial state $\ket{\Psi_{\mathrm{aug}}}$, not on the walker matrix $\Phi$, so the
one-shot estimator may be differentiated with respect to $\Phi$. For a
single shadow, define
\begin{equation}
    \label{eq:H_X}
    H_X
    =
    S^\dagger X\Phi.
\end{equation}

When $H$ is nonsingular, Jacobi's formula gives
\begin{equation}
\label{eq:first_derivative}
\partial_{\lambda}
\widehat D_{u,\vec z}(\Phi(\lambda))\big|_{\lambda=0}
=
\widehat D_{u,\vec z}(\Phi)\tr{H^{-1}H_X}.
\end{equation}

The derivative numerator and overlap denominator must be averaged separately.
Using $N_{\mathrm{shadow}}$ shadow records, we estimate
\begin{equation}
\label{eq:ratio_after_averaging}
\partial_{\lambda}\log D(\Phi(\lambda))\big|_0
\approx
\frac{\widehat D_X}{\widehat D},
\end{equation}
where
\begin{equation}
\label{eq:averaged_DX_D}
\begin{aligned}
\widehat D_X
&=
\frac{1}{N_{\mathrm{shadow}}}\sum_t
\partial_{\lambda}
\widehat D_{u^{(t)},\vec z^{(t)}}(\Phi(\lambda))\big|_0,
\\
\widehat D
&=
\frac{1}{N_{\mathrm{shadow}}}\sum_t
\widehat D_{u^{(t)},\vec z^{(t)}}(\Phi).
\end{aligned}
\end{equation}
The ratio of sample means is consistent as $N_{\mathrm{shadow}}\rightarrow\infty$ when
$D(\Phi)\neq0$, but it is generally biased at finite $N_{\mathrm{shadow}}$ and becomes
ill-conditioned near an overlap zero. Averaging the ratio obtained from each
individual shadow would estimate a different quantity and should not be used.

For local-energy two-body terms, use two generators $X$ and $Y$ and set
\begin{equation}
\label{eq:Phi_two_parameters}
\Phi(\lambda_1,\lambda_2)
=e^{\lambda_1 X}e^{\lambda_2 Y}\Phi.
\end{equation}
With
\begin{equation}
\label{eq:H_Y_H_XY}
H_Y=S^{\dagger}Y\Phi,
\qquad
H_{XY}=S^{\dagger}XY\Phi,
\end{equation}
Jacobi's formula yields
\begin{equation}
\label{eq:second_derivative}
\begin{split}
\partial_{\lambda_1}\partial_{\lambda_2}
\widehat D_{u,\vec z}(\Phi(\lambda_1,\lambda_2))\big|_0
&=
\widehat D_{u,\vec z}(\Phi)
\big[
\tr{H^{-1}H_{XY}}
\\
&\quad
-\tr{H^{-1}H_YH^{-1}H_X}
\\
&\quad
+\tr{H^{-1}H_X}\tr{H^{-1}H_Y}
\big].
\end{split}
\end{equation}
For Cholesky-vector local-energy terms one sets $X=Y=X_{\gamma}$, so the
linear solve for $H^{-1}H_X$ is reused for both first-derivative factors. 

\subsection{Variance Bound of Orbital-Rotation Shadows}
In classical shadows, the number of samples needed to estimate observables to an additive error $\epsilon$ depends on the shadow variance. For complex overlaps, Wan \textit{et al.}~\cite{wan2023matchgate} provide an explicit finite-size variance bound; the $\mathcal{O}(\sqrt{N}\log N)$ asymptotic bound is proved for Gaussian-projector expectation values and numerically suggested for the nonzero-particle overlap case. 
Low~\cite{low24_08964} showed that the averaged shadow variance of $D_{\vec q}^{\vec a}$ is bounded by a constant $\mathcal{O}(1)$. Although it's tempting to think that this indicates the number of samples of orbital rotation shadows scales as $\mathcal{O}(1)$, one needs to note an important subtlety: the average variance is taken over all of the $\eta$-RDMs, whereas only a subset of the $\eta$-RDMs correspond to the estimation of a Slater determinant overlap, as pointed out by Wan \textit{et al.}~\cite{wan2023matchgate}.  

Recently, West {\it et al.}~\cite{Larocca26_27254} showed that the orbital-rotation shadows for estimating squared overlaps $|\braket{\phi_i}{\Psi_T}|^2$ scales as $\mathcal{O}(\eta~\log~\eta)$, independent of the number of modes $N$. This is a marked contrast with matchgate shadows, whose known variance bounds for the same quantities grow with $N$, see ~\cite{wan2023matchgate}. However, QC-AFQMC requires the complex overlap $\braket{\phi_i}{\Psi_T}$ itself, including its phase. Its estimation via the augmented-ancilla protocol of Sec.~\ref{sec:augmented_observable} has, to our knowledge, not been analyzed before. Here we close this gap: we show that the single-record estimator is unbiased and that its complex variance is bounded above by $4\eta + 3$, independent of $N$, so that $O(\eta/\epsilon^2)$ records suffice. We also show (Appendix~\ref{app:ov-variance}, eq.~(\ref{eq:ov-lower-bound})) that the variance is bounded below by $\eta+1$ for every normalized trial–Slater-determinant pair, so the $\Theta(\eta)$ scaling is tight.

\begin{theorem}[Variance bound for overlap estimation]
\label{thm:ov-variance}
Let $1\leq\eta\leq N$, $N_{\mathrm{aug}}=N+\eta$, and let
$|\Psi_T\rangle$ and $|\phi\rangle=U_\eta(w)|\vec q\rangle$
be a normalized $\eta$-particle trial state and a normalized Slater
determinant, respectively. Prepare
$|\Psi_{\mathrm{aug}}\rangle=(|\Psi_T\rangle+|\vec a\rangle)/\sqrt{2}$,
where $\vec a=N+[\eta]$ and the physical trial state has empty ancillary
modes. Sample $u$ from complex Haar measure on
$\mathcal U_{N_{\mathrm{aug}}}$ and $\vec z$ from the ideal Born
distribution after applying $U_\eta(u)$.
Then the single-record estimator $\widehat O_{\vec q,w}$ is unbiased for
$o_\phi=\langle\phi|\Psi_T\rangle$ and obeys
\begin{align}
 \operatorname{Var}_{\mathbb C}(\widehat O_{\vec q,w})
 &=\mathbb E_{u,\vec z}
 \bigl|\widehat O_{\vec q,w}-o_\phi\bigr|^2
 \notag\\
 &\leq 4(\eta+1)\frac{N+2}{N+\eta+2}-1
 \leq4\eta+3.
 \label{eq:ov-tight-bound}
\end{align}
The first bound is attained when
$|\Psi_T\rangle=e^{i\theta}|\phi\rangle$.
\end{theorem}
The proof and exact state-dependent variance are given in
Appendix~\ref{app:ov-variance}.
For a normalized target Slater determinant fixed independently of the measurement
records, let $\widehat o_\phi$ be the ordinary mean of the single-record
overlap estimators evaluated on $N_{\mathrm{shadow}}$ independent records
generated under the ideal protocol of Theorem~\ref{thm:ov-variance}.
Independence and unbiasedness give
\begin{equation}
 \mathbb E|\widehat o_\phi-o_\phi|^2
 \leq\frac{4\eta+3}{N_{\mathrm{shadow}}}.
 \label{eq:ov-main-mean-mse}
\end{equation}
Thus, for $\epsilon>0$, a sufficient sample count for additive
root-mean-square overlap error at most $\epsilon$ is
\begin{equation}
 N_{\mathrm{shadow}}=
 \left\lceil\frac{4\eta+3}{\epsilon^2}\right\rceil.
 \label{eq:ov-main-rms-budget}
\end{equation}
For simultaneous high-confidence estimates of $L$ normalized target
Slater determinants fixed independently of the records, coordinatewise
median-of-means aggregation achieves complex absolute error at most
$\epsilon$ for every target with probability at least $1-\delta$, using
$\mathcal{O}(\eta\epsilon^{-2}\log(2L/\delta))$ independent records shared
by all targets, for $0<\delta<1$; see Appendix~\ref{app:ov-variance}.
This logarithmic confidence dependence requires median-of-means
aggregation; the ordinary mean obeys the mean-square guarantee above.
These are additive guarantees for fixed overlaps; final AFQMC energy
accuracy requires a separate analysis.

\subsection{Cost Scaling of QC-AFQMC}
\label{sec:scaling}
We now analyze the cost scaling for overlap evaluation. Constructing the matrix $A_i^{(t)}$ costs $\mathcal{O}(N\eta^2)$. The remaining operations involve only $\eta\times\eta$ matrices, including matrix multiplications, matrix determinants, and linear solves, each of which costs $\mathcal{O}(\eta^3)$. Therefore, the overall cost scales as
\begin{equation}
\label{eq:measured_row_direct_cost}
\mathcal{O}(N\eta^2+\eta^3)\sim\mathcal{O}(N\eta^2),
\end{equation}
per shadow and per walker. This improves the $\mathcal{O}(N^4)$ cost of the original matchgate-shadow overlap algorithm~\cite{wan2023matchgate,huang2024evaluating,zhao2026matchgate}. At fixed filling, $\eta\propto N$, the orbital-rotation cost is cubic, matching the recent matchgate algorithm~\cite{nappi2026qcafqmc}, while retaining explicit dependence on the particle number. 
For a fixed normalized target Slater determinant and additive root-mean-square overlap error $\epsilon$, combining Eq.~\ref{eq:measured_row_direct_cost} with the sufficient budget in Eq.~\ref{eq:ov-main-rms-budget} gives total arithmetic work $\mathcal{O}(N\eta^3/\epsilon^2)$, or $\mathcal{O}(N^4/\epsilon^2)$ at fixed filling. In each method, total overlap work is the per-shadow cost multiplied by $N_{\mathrm{shadow}}$ and the number of target Slater determinants. These fixed-overlap guarantees do not specify the measurement budget needed for a final AFQMC energy accuracy.

The cost scaling for computing force biases and local energies can be obtained in a similar way. For force biases in Eqs. \ref{eq:first_derivative}, $H^{-1}S^\dagger$ could be computed first with linear solve without forming $H^{-1}$ explicitly. The linear solve costs $\mathcal{O}(N\eta^2)$ and this is the only shadow-dependent step. The resulting matrix has size $\eta\times N$, we then accumulate over shadows, and contract it with $\Phi$ (size $N\times\eta$), which costs $\mathcal{O}(N^2\eta)$. The last step is to compute the trace, which can be done with cost $\mathcal{O}(N^2)$. The overall cost scaling for computing force bias is 
\begin{equation}
    \mathcal{O}(N_{\mathrm{shadow}}N\eta^2+N^2\eta+N_\gamma N^2)\sim\mathcal{O}(N\eta^3)
\end{equation}
in which we assume $N_\gamma\sim\mathcal{O}(N)$ and a fixed $\epsilon$. Force biases have the same cost scaling as overlaps. 

Local energies can be computed similarly to force bias, with the most expensive part being the quadratic term in the right hand of Eqs. \ref{eq:second_derivative}. One now has to explicitly form the matrix $H^{-1}S^\dagger X\Phi$, which costs $\mathcal{O}(N_{\mathrm{shadow}}N_\gamma N\eta^2)\sim\mathcal{O}(N^2\eta^3)$. 

\subsection{Virtual Correlation Energy}
One major advantage of QC-AFQMC is that it can account for both static and dynamic correlation. Static correlation is treated by selecting an active space and preparing a quantum trial state in that active space. Dynamic correlation is incorporated through the virtual correlation energy (VCE) technique~\cite{huggins2022unbiasing,zhao2026matchgate}.
Combining the two treatments
allows us to propagate walkers in the full orbital space. VCE approximates the full-space trial state as
\begin{equation}
    \ket{\Psi_T}=\ket{\Xi^\alpha_c}\otimes\ket{\Xi^\beta_c}\otimes\ket{\Psi_{T,a}}\otimes\ket{0^\alpha_v}\otimes\ket{0^\beta_v},
\end{equation}
where $\ket{\Psi_{T,a}}$ is the trial state prepared on a quantum computer within the active space with $\eta$ 
electrons, $\ket{\Xi^{\alpha (\beta)}_c}$ is a Slater determinant containing all frozen occupied $\alpha$ ($\beta$) orbitals with $\eta^{\alpha (\beta)}_c$ electrons, and $\ket{0^{\alpha (\beta)}_v}$ is the vacuum state in the space
of virtual $\alpha$ ($\beta$) orbitals, with $\eta=\eta_a^\alpha+\eta_a^\beta$. 

Given a walker in the full orbital space, VCE computes its overlap with the full-space trial state as
\begin{equation}
  \braket{\Psi_T}{\phi}
  =
  \frac{
  \det(L^\alpha_c)\det(L^\beta_c)
  }{
  \det(Q)
  }
  \braket{\Psi_{T,a}}{\tilde{\phi}_a}.
\end{equation}
where $L_c^{\alpha(\beta)}$ and $Q$ are matrices computed classically, and $\ket{\tilde{\phi}_a}$ is a rotated walker in the active space.

The active-space overlap $\braket{\Psi_{T,a}}{\tilde{\phi}_a}$ is then evaluated using orbital-rotation shadows and multiplied by the classically computed prefactor. The full VCE derivation is given in the appendix.
As one could see, the cost of VCE is dominated by QR decomposition of matrices with dimension of $(\eta^{\alpha(\beta)}_c+\eta^{\alpha(\beta)}_a)\times\eta^{\alpha(\beta)}_c$ , and computing matrix determinants of $\eta^{\alpha(\beta)}_c\times\eta^{\alpha(\beta)}_c$ and $\eta^{\alpha(\beta)}_a\times\eta^{\alpha(\beta)}_a$ matrices. Therefore, the VCE part of the cost for computing overlaps is (dropping the $\alpha(\beta)$ superscripts for simplicity)
\begin{equation}
    \mathcal{O}((\eta_c+\eta_a)\eta^2_c+\eta_c^3+\eta_a^3)\sim\mathcal{O}((\eta_c+\eta_a)\eta_c^2)
\end{equation}
The total cost for computing overlaps with VCE becomes
\begin{equation}
    \mathcal{O}((\eta_c+\eta_a)\eta_c^2 + N_a\eta_a^3/\epsilon^2)
\end{equation}
per walker, in which $N_a$ is the total number of modes in the active space.   

Compared with previous matchgate-shadow QC-AFQMC implementations~\cite{huang2024evaluating,zhao2026matchgate}, this work significantly improves the efficiency of VCE from the prefactor perspective. First, we replace the expensive singular value decomposition and its derivatives with cheaper QR decompositions. Second, unlike the previous matchgate-shadow implementations, which orthonormalize the walkers and embed them into a real orthogonal matrix, orbital-rotation shadows accept arbitrary non-orthonormal walkers. This allows us to further skip two additional QR decompositions and their derivatives in the matchgate-shadow VCE algorithm. 

\section{Results}
\label{sec:results}
\subsection{Overlap Errors}
As a first validation of orbital-rotation shadows, we evaluate overlaps between a quantum state and arbitrary non-orthogonal Slater determinants. We choose two simple systems: linear hydrogen chains with 4 and 6 hydrogen atoms, separated by 1\AA. For both systems, we prepare a unitary pair coupled cluster double (UpCCD) state, where instead of optimizing the state parameters, we take amplitudes from the Moller-Plesset second order perturbation theory (MP2). We then apply randomized orbital rotation unitaries to the state and perform one-shot measurements. The measured classical shadows are then used to compute the state's overlap with 100 randomly sampled Slater determinants. We also compute the exact overlaps by forming the $2^N$-dimensional state vector of the prepared quantum state and computing the inner product with the Slater determinant's state vector. The mean absolute error of the shadow computed overlaps with respect to the exact overlaps is shown in Figure \ref{fig:overlap_errors}. As expected, the amount of error reduces as the number of shadows increases. 
\begin{figure}[tbh!]
    \centering
    \includegraphics[width=\columnwidth]{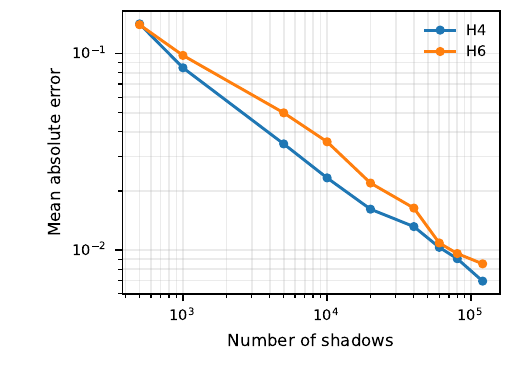}
     \caption{Mean absolute overlap error from orbital-rotation shadows for hydrogen chains with 4 and 6 atoms at a spacing of 1\AA. The trial state is the UpCCD state with MP2 parameters. The overlaps are computed with respect to 100 randomly sampled Slater determinants.}
    \label{fig:overlap_errors}
\end{figure}

\subsection{QC-AFQMC with Orbital Rotation Shadows}
\begin{figure}[tbh!]
    \centering
    \includegraphics[width=\columnwidth]{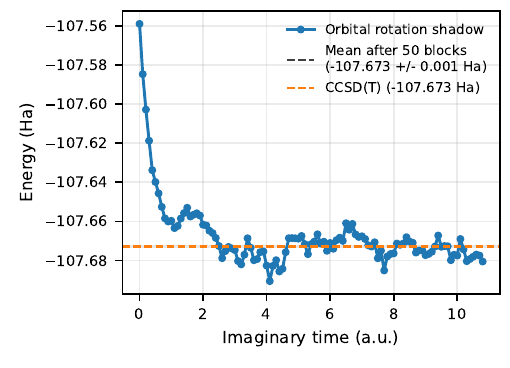}
     \caption{QC-AFQMC energy trajectory for N$_2$ using orbital-rotation shadows with an UpCCD trial state initialized from MP2 parameters. The dashed horizontal line indicates the reference CCSD(T) energy.}
    \label{fig:n2_qc_afqmc}
\end{figure}
Having validated orbital-rotation shadows at the overlap level, we next demonstrate their applicability in QC-AFQMC. As a proof-of-concept calculation, we simulate the N$_2$ molecule in the STO-3G basis with a (6e, 6o) active space. The UpCCD state with MP2 parameters is used as the active-space trial state, and correlation outside the active space is included with VCE. We use 1024 walkers and 20,000 orbital-rotation shadows, and propagate the walkers for 1080 steps with a time step of 0.01~a.u. The comparison with the CCSD(T) energy is shown in \Cref{fig:n2_qc_afqmc}. The converged QC-AFQMC energy obtained with orbital-rotation shadows 
is in agreement with that of CCSD(T) within an error bar of 0.001~Ha.

\subsection{Scaling to Large Systems}
Compared with our earlier matchgate-shadow implementation~\cite{zhao2026matchgate}, orbital-rotation shadows reduce the asymptotic cost of overlap evaluation and simplify post-processing. This efficiency makes it possible to scale QC-AFQMC to larger systems. To illustrate this, we consider a series of $\pi$-conjugated molecules: naphthalene (C$_{10}$H$_8$), anthracene (C$_{14}$H$_{10}$), tetracene (C$_{18}$H$_{12}$), and pentacene (C$_{22}$H$_{14}$). These $\pi$-conjugated systems are high-performance organic semiconductors and are widely used in organic light-emitting diodes and organic photovoltaic cells~\cite{anthony2006acenes}.
To model these molecules, we select an active space containing all $\pi$ electrons. For a molecule with a chemical formula C$_{4n+2}$H$_{2n+4}$, where $n$ is the number of fused benzene rings, this gives an active space of $4n+2$ electrons in $4n+2$ orbitals. We use the cc-pVDZ basis and freeze $1s$-orbitals on the carbon atoms. 
\begin{table*}[t]
\centering
\caption{Predicted QC-AFQMC post-processing runtimes for $N_w=1024$ walkers and $N_{\mathrm{shadow}}=65,536$ shadows for a fixed propagation schedule of 150 blocks and 10 time steps per block (in the last two columns).}
\label{tab:runtime_prediction}
\begin{tabular}{lcccc}
\hline
System
& \multicolumn{1}{c}{Qubits}
& \multicolumn{1}{c}{Time per block (min)}
& \multicolumn{1}{c}{8$\times$ A100 (hr)}
& \multicolumn{1}{c}{320$\times$ H200 (min)} \\
\hline
Naphthalene CAS(10e,10o) & 30 & 14.77 & 36.92 & 39.56 \\
Anthracene CAS(14e,14o)  & 42 & 15.71 & 39.27 & 42.07 \\
Tetracene CAS(18e,18o)   & 54 & 17.85 & 44.63 & 47.82 \\
Pentacene CAS(22e,22o)   & 66 & 21.98 & 54.96 & 58.89 \\
\hline
\end{tabular}
\end{table*}

After selecting the active space, we measure the amount of time it takes to run one block with ten steps of QC-AFQMC, with a various number of walkers and shadows. The runtime is measured on 8$\times$NVIDIA A100 GPUs. 
For these timing benchmarks, we generate pseudo-shadows by applying randomized orbital rotation unitaries to the Hartree-Fock (HF) state, which simply turns it to another Slater determinant. We then sample from the new Slater determinant, which could be done efficiently on classical computers.  

After obtaining the runtime data, we fit a performance model as a linear combination of
physically motivated scaling features. For each benchmark point, the input
variables are the number of active-space orbitals $N_a$, the number of active
electrons $\eta$, the number of occupied orbitals outside the active space $N_o$,
the number of Cholesky vectors $N_{\gamma}$, the number of walkers $N_w$,
and the number of orbital-rotation shadows $N_{\mathrm{shadow}}$.

The fitted model is
\begin{equation}
\begin{split}
    &T_{\mathrm{block}}
=
c_0
+ c_w N_w
+ c_{ws} N_w N_{\mathrm{shadow}} \\
&+ c_{\mathrm{est}} N_{\gamma} N_w N_{\mathrm{shadow}} N_a \eta^2
+ c_{\mathrm{occ}} N_{\gamma} N_w N_o^2 (N_o + N_a).
\end{split}
\end{equation}

Here $T_{\mathrm{block}}$ is the measured wall time for one AFQMC block on the
reference hardware configuration.  The term
$N_{\gamma}N_wN_{\mathrm{shadow}}N_a\eta^2$ captures the dominant orbital-rotation
shadow estimator cost, which scales with the number of Cholesky vectors, walkers,
shadows, active-space orbitals, and active-space electron pairs.  The term
$N_{\gamma}N_wN_o^2(N_o+N_a)$ models the cost associated with processing
the occupied orbitals outside the active space using VCE.  The additional $N_wN_{\mathrm{shadow}}$ term captures
non-asymptotic per-walker/per-shadow overhead, such as data movement and batching
costs, that is visible in the measured data.

For each timing measurement $i$, we construct the feature vector
\begin{equation}
\begin{split}
\mathbf{x}_i =
\big[
&1,\,
N_w^i,\,
N_w^i N_{\mathrm{shadow}}^i,\,
N_{\gamma}^i N_w^i N_{\mathrm{shadow}}^i N_a^i \eta_i^2, \\
&N_{\gamma}^i N_w^i (N_o^i)^2(N_o^i + N_a^i)
\big].
\end{split}
\end{equation}

The coefficients are obtained by ordinary least squares:
\begin{equation}
\mathbf{c}
=
\operatorname*{arg\,min}_{\mathbf{c}}
\sum_i
\left(
T_i - \mathbf{x}_i \cdot \mathbf{c}
\right)^2,
\end{equation}
where
\begin{equation}
\mathbf{c}
=
\left[
c_0,\,
c_w,\,
c_{ws},\,
c_{\mathrm{est}},\,
c_{\mathrm{occ}}
\right].
\end{equation}

After fitting, the model predicts the runtime of a new system by substituting
its values of $(N_a, \eta, N_o, N_{\gamma}, N_w, N_{\mathrm{shadow}})$ into the fitted expression for
$T_{\mathrm{block}}$. 

\begin{figure*}[tbh!]
    \centering
    \includegraphics[width=0.95\textwidth]{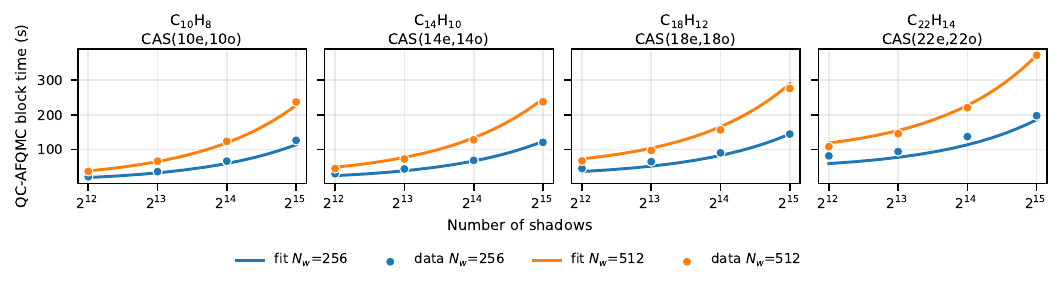}
    \caption{Measured and fitted QC-AFQMC times for 1 block with 10 steps for orbital-rotation shadows across the $\pi$-conjugated benchmark series. Points show measured timings and solid lines show the fitted physical scaling model as a function of the number of shadows for different walker counts.}
    \label{fig:performance_model_fit}
\end{figure*}

\Cref{fig:performance_model_fit} shows the measured timings and fitted curves used to calibrate the performance model and Table~\ref{tab:runtime_prediction} shows the corresponding predicted QC-AFQMC runtimes for the four larger $\pi$-conjugated systems using 65,536 orbital-rotation shadows, 1024 walkers and 150 blocks with 10 steps per block. We assume perfect strong scaling with the number of GPUs and that an NVIDIA H200 is 1.4$\times$ faster than an NVIDIA A100. Under these assumptions, pentacene with a (22e, 22o) active space is predicted to require slightly more than two days on 8 NVIDIA A100 GPUs and less than an hour on 320 NVIDIA H200 GPUs.

To illustrate the significance of performance improvements, we turn to a comparison of runtime projections derived here with a projection that draws on our previous study~\cite{zhao2026matchgate}. We choose Cu$_2$O$_2$(NH$_3$)$_2^{2+}$ with a (24e, 24o) active space as an illustrative example. The previous algorithm would require approximately 1~s on one NVIDIA H200 GPU to execute one shadow and one walker for one block (ten force-bias evaluations and one local-energy evaluation), which translates to approximately one year on 320 NVIDIA H200 GPUs with 1024 walkers,  65,536 shadows and 150 blocks. The corresponding projected runtime with the orbital-rotation-shadow algorithm for the similarly sized pentacene molecule is less than an hour (Table~\ref{tab:runtime_prediction}).

All costs discussed here are classical post-processing costs. To estimate the full end-to-end cost of QC-AFQMC with orbital-rotation shadows, one must also include the quantum shadow measurement time. This measurement time is harder to estimate because it depends strongly on the choice of trial state and on hardware gate speeds. For example, in our previous study, measuring nearly 60,000 matchgate shadows for an UpCCD trial state on a 16-qubit system required approximately one day on the IonQ Forte QPU. For orbital-rotation shadows, the quantum execution cost is expected to be somewhat higher because of the additional ancilla qubits and the non-Clifford nature of the measurement circuits. Future studies should therefore focus on reducing quantum execution time to further accelerate QC-AFQMC.

\section{Conclusions}
\label{sec:conclusions}
In this work, we derived an efficient algorithm for estimating overlaps between arbitrary Slater determinants and quantum states using particle-number-conserving fermionic shadow tomography, or orbital-rotation shadows. We also presented the derivatives required to compute force biases and local energies in QC-AFQMC. We provided proof-of-concept demonstrations for overlap estimation and for an end-to-end QC-AFQMC workflow with orbital-rotation shadows. We also estimated runtimes by benchmarking the algorithm on a series of $\pi$-conjugated systems. Compared with our earlier implementation of matchgate shadows, orbital-rotation shadows provide dramatic improvements in post-processing efficiency. For pentacene with a (22e, 22o) active space in the cc-pVDZ basis, orbital-rotation-shadow QC-AFQMC has a projected post-processing time of less than an hour on the large-GPU configuration considered here.

Beyond efficiency, orbital-rotation shadows open new opportunities. 
For instance, molecular orbitals have additional symmetries, such as spin and point-group symmetries. It is therefore natural to ask whether these symmetries can be used to sample only shadows that contribute significantly to the target observables, thereby reducing the total number of required shadows.

Admittedly, orbital-rotation shadows require additional quantum resources. For an $N$-mode, $\eta$-fermion problem, the construction used here adds $\eta$ ancilla qubits, bringing the total to $N+\eta$, compared with $N$ qubits for matchgate shadows. Low~\cite{low24_08964} proposed reducing this overhead to a single ancilla. However, whether this construction retains the simple overlap and derivative formulas developed here remains to be investigated.

The measurement circuits also differ. The Haar-random orbital rotations used here are dense and generally require arbitrary-angle rotation gates. Although both approaches admit circuit decompositions with $\mathcal{O}(N)$ depth, matchgate shadows can use Clifford measurement ensembles and exploit their sparse signed-permutation structure and fixed rotation angles. Heyraud \textit{et al.}~\cite{heyraud2025unified} also showed how perfect-matching representations can be applied to further reduce gate counts to optimality while preserving the measurement channel and estimator variance. In some cases, these simplifications may make matchgate shadows less expensive to execute than the orbital-rotation circuits considered here, however, the relative cost utlimately depends on the native gate set, connectivity, and compilation.

This work lays the foundation for applying orbital-rotation shadows to efficiently evaluate overlaps between a quantum state and an arbitrary Slater determinant. In QC-AFQMC, where these overlaps are evaluated for every walker at every step, this removes classical post-processing as the dominant cost, leaving quantum execution time for shadow collection as the quantity to reduce. This could be done through more compact circuit design, more efficient fermion-to-qubit mapping, improved sampling strategies, and advanced circuit compilation approaches tailored to a specific qubit modality and architecture. Future work should examine these directions, along with hardware demonstrations of orbital-rotation shadow QC-AFQMC on large strongly correlated systems (e.g., pentacene) beyond the reach of classical approaches. 

\makeatletter
\immediate\write\@auxout{\string\citation{apsrev42Control}}
\makeatother
\bibliographystyle{apsrev4-2}
\bibliography{Journal_Short_Name, biblio, ionq}

@string{ CP = "Chem. Phys." }

@string{ CR = "Chem. Rev." }

@string{ JCP = "J. Chem. Phys." }

@string{ JCTC = "J. Chem. Theory Comput." }

@string{ ME = "Method. Enzymol." }

@string{ N = "Nature" }

@string{ NC = "Nat. Chem." }

@string{ NP = "Nat. Phys." }

@string{ NPJQ = "npj Quantum Inf." }

@string{ NRP = "Nat. Rev. Phys."}

@string{ PR = "Phys. Rev." }

@string{ PRB = "Phys. Rev. B" }

@string{ PRX = "Phys. Rev. X" }

@string{ PRXQ = "PRX Quantum" }

@string{ PRL = "Phys. Rev. Lett." }

@string{ PRR = "Phys. Rev. Research"}

@string{ PRp = "Phys. Rep." }

@string{ PNAS = "Proc. Natl. Acad. Sci. U.S.A." }

@string{ QST = "Quantum Sci. Technol." }

@string{ RMP = "Rev. Mod. Phys." }

@string{ S = "Science" }

@string{ WCMS = "WIREs Comput. Mol. Sci." }

@CONTROL{apsrev42Control,
    author = "08",
    editor = "1",
    pages = "0",
    title = "0",
    year = "1"
  }

@article{heyraud2025unified,
  author = {Heyraud, Valentin and Chomet, H{\'e}loise and Tilly, Jules},
  title = {{Unified framework for matchgate classical shadows}},
  journal = {npj Quantum Information},
  volume = {11},
  number = {1},
  pages = {65},
  year = {2025},
  doi = {10.1038/s41534-025-01015-y},
}

@article{aspuru2005simulated,
  author = {Aspuru-Guzik, Alan and Dutoi, Anthony D. and Love, Peter J. and Head-Gordon, Martin},
  title = {{Simulated Quantum Computation of Molecular Energies}},
  journal = S,
  volume = {309},
  number = {5741},
  pages = {1704--1707},
  year = {2005},
}

@article{peruzzo2014variational,
  author = {Peruzzo, Alberto and McClean, Jarrod and Shadbolt, Peter and Yung, Man-Hong and Zhou, Xiao-Qi and Love, Peter J. and Aspuru-Guzik, Alan and O'Brien, Jeremy L.},
  title = {{A Variational Eigenvalue Solver on a Photonic Quantum Processor}},
  journal = NC,
  volume = {5},
  pages = {4213},
  year = {2014},
}

@article{mcclean2018barren,
  author = {McClean, Jarrod R. and Boixo, Sergio and Smelyanskiy, Vadim N. and Babbush, Ryan and Neven, Hartmut},
  title = {{Barren Plateaus in Quantum Neural Network Training Landscapes}},
  journal = NC,
  volume = {9},
  pages = {4812},
  year = {2018},
}

@article{huang2020predicting,
  author = {Huang, Hsin-Yuan and Kueng, Richard and Preskill, John},
  title = {{Predicting Many Properties of a Quantum System from Very Few Measurements}},
  journal = NP,
  volume = {16},
  pages = {1050--1057},
  year = {2020},
}

@article{zhang2003phaseless,
  author = {Zhang, Shiwei and Krakauer, Henry},
  title = {{Quantum Monte Carlo Method Using Phase-Free Random Walks with Slater Determinants}},
  journal = PRL,
  volume = {90},
  pages = {136401},
  year = {2003},
}

@article{lee2022twenty,
  author = {Lee, Joonho and Pham, Hung Q. and Reichman, David R.},
  title = {{Twenty Years of Auxiliary-Field Quantum Monte Carlo in Quantum Chemistry: An Overview and Assessment on Main Group Chemistry and Bond-Breaking}},
  journal = JCTC,
  volume = {18},
  number = {12},
  pages = {7024--7042},
  year = {2022},
}

@article{huggins2022unbiasing,
  author = {Huggins, William J. and O'Gorman, Bryan A. and Rubin, Nicholas C. and Reichman, David R. and Babbush, Ryan and Lee, Joonho},
  title = {{Unbiasing Fermionic Quantum Monte Carlo with a Quantum Computer}},
  journal = N,
  volume = {603},
  number = {7901},
  pages = {416--420},
  year = {2022},
}

@article{robledo2024chemistry,
  author = {Robledo-Moreno, Javier and Motta, Mario and Haas, Holger and Javadi-Abhari, Ali and Jurcevic, Petar and Kirby, William and Martiel, Simon and Sharma, Kunal and Sharma, Sandeep and Shirakawa, Tomonori and others},
  title = {{Chemistry Beyond Exact Solutions on a Quantum-Centric Supercomputer}},
  journal = {arXiv preprint},
  year = {2024},
  eprint = {2405.05068},
  archivePrefix = {arXiv},
  primaryClass = {quant-ph},
}

@article{wan2023matchgate,
  author = {Wan, Kianna and Huggins, William J. and Lee, Joonho and Babbush, Ryan},
  title = {{Matchgate Shadows for Fermionic Quantum Simulation}},
  journal = CMP,
  volume = {404},
  number = {2},
  pages = {629--700},
  year = {2023},
}

@article{low24_08964,
  author = {Low, Guang Hao},
  title = {{Classical Shadows of Fermions with Particle Number Symmetry}},
  journal = {arXiv preprint},
  year = {2024},
  eprint = {2208.08964},
  archivePrefix = {arXiv},
  primaryClass = {quant-ph}
}

@article{anthony2006acenes,
  author = {Anthony, John E.},
  title = {{Functionalized Acenes and Heteroacenes for Organic Electronics}},
  journal = CR,
  volume = {106},
  number = {12},
  pages = {5028--5048},
  year = {2006},
}

@article{hachmann2007radical,
  author = {Hachmann, Johannes and Dorando, Juan J. and Avil{\'e}s, Michael and Chan, Garnet Kin-Lic},
  title = {{The Radical Character of the Acenes: A Density Matrix Renormalization Group Study}},
  journal = JCP,
  volume = {127},
  pages = {134309},
  year = {2007},
}

@article{keimer2015quantum,
  author = {Keimer, B. and Kivelson, S. A. and Norman, M. R. and Uchida, S. and Zaanen, J.},
  title = {{From Quantum Matter to High-Temperature Superconductivity in Copper Oxides}},
  journal = N,
  volume = {518},
  pages = {179--186},
  year = {2015},
}

@article{lee2006doping,
  author = {Lee, Patrick A. and Nagaosa, Naoto and Wen, Xiao-Gang},
  title = {{Doping a Mott Insulator: Physics of High-Temperature Superconductivity}},
  journal = RMP,
  volume = {78},
  pages = {17--85},
  year = {2006},
}

@article{reiher2017elucidating,
  author = {Reiher, Markus and Wiebe, Nathan and Svore, Krysta M. and Wecker, Dave and Troyer, Matthias},
  title = {{Elucidating Reaction Mechanisms on Quantum Computers}},
  journal = PNAS,
  volume = {114},
  number = {29},
  pages = {7555--7560},
  year = {2017},
}

@article{abrams1999quantum,
  author = {Abrams, Daniel S. and Lloyd, Seth},
  title = {{Quantum Algorithm Providing Exponential Speed Increase for Finding Eigenvalues and Eigenvectors}},
  journal = PRL,
  volume = {83},
  number = {24},
  pages = {5162--5165},
  year = {1999},
}

@article{aaronson2020shadow,
  author = {Aaronson, Scott},
  title = {{Shadow Tomography of Quantum States}},
  journal = {SIAM Journal on Computing},
  volume = {49},
  number = {5},
  pages = {STOC18-368--STOC18-394},
  year = {2020},
}

@article{bauer2020quantum,
  author = {Bauer, Bela and Bravyi, Sergey and Motta, Mario and Chan, Garnet Kin-Lic},
  title = {{Quantum Algorithms for Quantum Chemistry and Quantum Materials Science}},
  journal = CR,
  volume = {120},
  number = {22},
  pages = {12685--12717},
  year = {2020},
}

@article{cao2019quantum,
  author = {Cao, Yudong and Romero, Jonathan and Olson, Jonathan P. and Degroote, Matthias and Johnson, Peter D. and Kieferov{\'a}, M{\'a}ria and Kivlichan, Ian D. and Menke, Tim and Peropadre, Borja and Sawaya, Nicolas P. D. and Sim, Sukin and Veis, Libor and Aspuru-Guzik, Al{\'a}n},
  title = {{Quantum Chemistry in the Age of Quantum Computing}},
  journal = CR,
  volume = {119},
  number = {19},
  pages = {10856--10915},
  year = {2019},
}

@article{cerezo2021variational,
  author = {Cerezo, M. and Arrasmith, Andrew and Babbush, Ryan and Benjamin, Simon C. and Endo, Suguru and Fujii, Keisuke and McClean, Jarrod R. and Mitarai, Kosuke and Yuan, Xiao and Cincio, Lukasz and Coles, Patrick J.},
  title = {{Variational Quantum Algorithms}},
  journal = NRP,
  volume = {3},
  pages = {625--644},
  year = {2021},
}

@article{farhi2014quantum,
  author = {Farhi, Edward and Goldstone, Jeffrey and Gutmann, Sam},
  title = {{A Quantum Approximate Optimization Algorithm}},
  journal = {arXiv preprint},
  year = {2014},
  eprint = {1411.4028},
  archivePrefix = {arXiv},
  primaryClass = {quant-ph}
}

@article{grimsley2019adaptive,
  author = {Grimsley, Harper R. and Economou, Sophia E. and Barnes, Edwin and Mayhall, Nicholas J.},
  title = {{An Adaptive Variational Algorithm for Exact Molecular Simulations on a Quantum Computer}},
  journal = NC,
  volume = {10},
  pages = {3007},
  year = {2019},
}

@article{hempel2018quantum,
  author = {Hempel, Cornelius and Maier, Christine and Romero, Jonathan and McClean, Jarrod and Monz, Thomas and Shen, Heng and Jurcevic, Petar and Lanyon, Ben P. and Love, Peter and Babbush, Ryan and Aspuru-Guzik, Al{\'a}n and Blatt, Rainer and Roos, Christian F.},
  title = {{Quantum Chemistry Calculations on a Trapped-Ion Quantum Simulator}},
  journal = PRX,
  volume = {8},
  pages = {031022},
  year = {2018},
}

@article{kandala2017hardware,
  author = {Kandala, Abhinav and Mezzacapo, Antonio and Temme, Kristan and Takita, Maika and Brink, Markus and Chow, Jerry M. and Gambetta, Jay M.},
  title = {{Hardware-Efficient Variational Quantum Eigensolver for Small Molecules and Quantum Magnets}},
  journal = N,
  volume = {549},
  pages = {242--246},
  year = {2017},
}

@article{kanno2023quantum,
  author = {Kanno, Keita and Kohda, Masaya and Imai, Ryosuke and Koh, Sho and Mitarai, Kosuke and Mizukami, Wataru and Nakagawa, Yuya O.},
  title = {{Quantum-Selected Configuration Interaction: Classical Diagonalization of Hamiltonians in Subspaces Selected by Quantum Computers}},
  journal = {arXiv preprint},
  year = {2023},
  eprint = {2302.11320},
  archivePrefix = {arXiv},
  primaryClass = {quant-ph},
}

@article{mcardle2020quantum,
  author = {McArdle, Sam and Endo, Suguru and Aspuru-Guzik, Al{\'a}n and Benjamin, Simon C. and Yuan, Xiao},
  title = {{Quantum Computational Chemistry}},
  journal = RMP,
  volume = {92},
  pages = {015003},
  year = {2020},
}

@article{romero2018strategies,
  author = {Romero, Jonathan and Babbush, Ryan and McClean, Jarrod R. and Hempel, Cornelius and Love, Peter J. and Aspuru-Guzik, Al{\'a}n},
  title = {{Strategies for Quantum Computing Molecular Energies Using the Unitary Coupled Cluster Ansatz}},
  journal = QST,
  volume = {4},
  number = {1},
  pages = {014008},
  year = {2018},
}

@article{omalley2016scalable,
  author = {O'Malley, P. J. J. and Babbush, Ryan and Kivlichan, Ian D. and Romero, Jonathan and McClean, Jarrod R. and Barends, R. and Kelly, J. and Roushan, P. and Tranter, A. and Ding, N. and Campbell, B. and Chen, Yu and Chen, Zijun and Chiaro, B. and Dunsworth, A. and Fowler, A. G. and Jeffrey, E. and Lucero, E. and Megrant, A. and Mutus, J. Y. and Neeley, M. and Neill, C. and Quintana, C. and Sank, D. and Vainsencher, A. and Wenner, J. and White, T. C. and Coveney, P. V. and Love, P. J. and Neven, H. and Aspuru-Guzik, Al{\'a}n and Martinis, J. M.},
  title = {{Scalable Quantum Simulation of Molecular Energies}},
  journal = PRX,
  volume = {6},
  pages = {031007},
  year = {2016},
}

@article{yu2025quantum,
  author = {Yu, Jeffery and Robledo-Moreno, Javier and Iosue, Joseph T. and Bertels, Luke and Claudino, Daniel and Fuller, Bryce and Groszkowski, Peter and Humble, Travis S. and Jurcevic, Petar and Kirby, William and Maier, Thomas A. and Motta, Mario and Pokharel, Bibek and Seif, Alireza and Shehata, Amir and Sung, Kevin J. and Tran, Minh C. and Tripathi, Vinay and Mezzacapo, Antonio and Sharma, Kunal},
  title = {{Quantum-Centric Algorithm for Sample-Based Krylov Diagonalization}},
  journal = {arXiv preprint},
  year = {2025},
  eprint = {2501.09702},
  archivePrefix = {arXiv},
  primaryClass = {quant-ph},
}

@article{tilly2022variational,
  author = {Tilly, Jules and Chen, Hongxiang and Cao, Shuxiang and Picozzi, Dario and Setia, Kanav and Li, Ying and Grant, Edward and Wossnig, Leonard and Rungger, Ivan and Booth, George H. and Tennyson, Jonathan},
  title = {{The Variational Quantum Eigensolver: A Review of Methods and Best Practices}},
  journal = PR,
  volume = {986},
  pages = {1--128},
  year = {2022},
}

@article{gonthier2022measurements,
  author = {Gonthier, J{\'e}r{\^o}me F. and Radin, Maxwell D. and Buda, Corneliu and Doskocil, Eric J. and Abuan, Clena M. and Romero, Jhonathan},
  title = {{Measurements as a Roadblock to Near-Term Practical Quantum Advantage in Chemistry: Resource Analysis}},
  journal = PRR,
  volume = {4},
  pages = {033154},
  year = {2022},
}

@article{zhang1997constrained,
  author = {Zhang, Shiwei and Carlson, J. and Gubernatis, J. E.},
  title = {{Constrained Path Monte Carlo Method for Fermion Ground States}},
  journal = PRB,
  volume = {55},
  number = {12},
  pages = {7464--7477},
  year = {1997},
}

@article{motta2018abinitio,
  author = {Motta, Mario and Zhang, Shiwei},
  title = {{Ab Initio Computations of Molecular Systems by the Auxiliary-Field Quantum Monte Carlo Method}},
  journal = WCMS,
  volume = {8},
  pages = {e1364},
  year = {2018},
}

@article{shee2023transformative,
  author = {Shee, James and Weber, John L. and Reichman, David R. and Friesner, Richard A. and Zhang, Shiwei},
  title = {{On the Potentially Transformative Role of Auxiliary-Field Quantum Monte Carlo in Quantum Chemistry: A Highly Accurate Method for Transition Metals and Beyond}},
  journal =JCP,
  volume = {158},
  pages = {140901},
  year = {2023},
}

@article{troyer2005complexity,
  author = {Troyer, Matthias and Wiese, Uwe-Jens},
  title = {{Computational Complexity and Fundamental Limitations to Fermionic Quantum Monte Carlo Simulations}},
  journal = PRL,
  volume = {94},
  pages = {170201},
  year = {2005},
}

@article{elben2023randomized,
  author = {Elben, Andreas and Flammia, Steven T. and Huang, Hsin-Yuan and Kueng, Richard and Preskill, John and Vermersch, Beno{\^i}t and Zoller, Peter},
  title = {{The Randomized Measurement Toolbox}},
  journal = NRP,
  volume = {5},
  pages = {9--24},
  year = {2023},
}

@article{huang2024evaluating,
  author = {Huang, Benchen and Chen, Yi-Ting and Gupt, Brajesh and Suchara, Martin and Tran, Anh and McArdle, Sam and Galli, Giulia},
  title = {{Evaluating a Quantum-Classical Quantum Monte Carlo Algorithm with Matchgate Shadows}},
  journal = PRR,
  volume = {6},
  pages = {043063},
  year = {2024},
}

@article{zhao2026matchgate,
  author = {Zhao, Luning and Goings, Joshua J. and Aboumrad, Willie and Arrasmith, Andrew and Calderin, Lazaro and Churchill, Spencer and Gabay, Dor and Harvey-Brown, Thea and Hiles, Melanie and others},
  title = {{Quantum-Classical Auxiliary Field Quantum Monte Carlo with Matchgate Shadows on Trapped Ion Quantum Computers}},
  journal = PRR,
  volume = {8},
  pages = {033061},
  year = {2026},
}

@article{nappi2026qcafqmc,
  title={Quantum-Classical Auxiliary-Field Quantum Monte Carlo at the Edge of Practicability}, 
  author={Francesco Nappi and Matthew Kiser and Fedor Šimkovic IV},
  journal = {arXiv preprint},
  year={2026},
  eprint={2606.19239},
  archivePrefix={arXiv},
  primaryClass={quant-ph},
  url={https://arxiv.org/abs/2606.19239}, 
}

@article{reinholdt2025critical,
  author = {Reinholdt, Peter and Ziems, Karl Michael and Kjellgren, Erik Rosendahl and Coriani, Sonia and Sauer, Stephan P. A. and Kongsted, Jacob},
  title = {{Critical Limitations in Quantum-Selected Configuration Interaction Methods}},
  journal =JCTC,
  volume = {21},
  number = {14},
  pages = {6811--6822},
  year = {2025},
}

@article{Zhao20_174105,
    author = {P. R. C. Kent and A. Annaberdiyev and A. Benali and M. C. Bennet and E. J. L. Borda and P. Doak and et al},
    journal = JCP,
    title = {{QMCPACK: Advances in the Development, Efficiency, and Application of Auxiliary Field and Real-Space Variational and Diffusion Quantum Monte Carlo}},
    volume = {152},
    pages = {174105},
    year = {2020},
}

@article{Gambetta17_23879,
    author = {A. Kandala and A. Mezzacapo and K. Temme and M. Takita and M. Brink and J. M. Chow and J. M. Gambetta},
    journal = N,
    title = {{Hardware-Efficient Variational Quantum Eigensolver for Small Molecules and Quantum Magnets}},
    volume = {549},
    pages = {242-246},
    year = {2017},
}

@article{Neven18_041015,
    author = {R. Babbush and C. Gidney and D. W. Berry and N. Wiebe and J. McClean and A. Paler and A. Fowler and H. Neven},
    title = {{Encoding Electronic Spectra in Quantum Circuits with Linear T Complexity}},
    journal = PRX,
    volume = {8},
    pages = {041015},
    year = {2018},
}

@article{Babbush21_030305,
    author = {J. Lee and D. W. Berry and C. Gidney and W. J. Huggins and J. R. McClean and N. Wiebe and R. Babbush},
    title = {{Even More Efficient Quantum Computations of Chemistry Through Tensor Hypercontraction}},
    journal = PRXQ, 
    volume = {2},
    pages = {030305},
    year = {2021},
}

@article{Zhao22_60,
  title = {{Orbital-Optimized Pair-Correlated Electron Simulations on Trapped-Ion Quantum Computers}},
  author = {L. Zhao and J. Goings and K. Wright and J. Nguyen and J. Kim and S. Johri and K. Shin and W. Kyoung and J. I. Fuks and J-K. K. Rhee and Y. M. Rhee},
  journal = NPJQ,
  volume = {9},
  number = {60},
  year = {2023},
}

@article{Umrigar143_164105,
  author = {C. J. Umrigar},
  title = {{Observations on Variational and Projector Monte Carlo Methods}},
  journal = JCP,
  volume = {143},
  pages = {164105},
  year = {2015},
}

@article{Miyake24_57,
  author = {A. Zhao and A. Miyake},
  title = {{Group-Theoretic Error Mitigation Enabled by Classical Shadows and Symmetries}},
  journal = NPJQ,
  volume = {10},
  pages = {57},
  year = {2024},
}

@article{Zhao21_110504,
  author = {A. Zhao and N. C. Rubin and A. Miyake},
  title = {{Fermionic Partial Tomography via Classical Shadows}},
  journal = PRL,
  volume = {127},
  pages = {110504},
  year = {2021},
}

@article{Larocca26_27254,
  author = {M. West and M. Cerezo and M. Larocca},
  title = {{Particle-Preserving Fermionic Shadows with Mode-Independent Sample Complexity}},
  journal = {Preprint at https://doi.org/10.48550/arXiv.2606.27254},
  year = {2026},
  }

@article{Koonin97_1,
title = {{Shell model Monte Carlo methods}},
journal = PRp,
volume = {278},
number = {1},
pages = {1-77},
year = {1997},
issn = {0370-1573},
author = {S.E. Koonin and D.J. Dean and K. Langanke},
}

@article{Motta26_01138,
  title = {{Crossing the 12,000-Atom Barrier with Heterogeneous Quantum-Classical Supercomputing: Quantum Chemistry of Protein-Ligand Complexes}},
  author = {K. M. Merz Jr and A. Shajan and D, Kaliakin and F. Liang and Y. Otsuka and T. Shirakawa and L. Broers and H. Xu and M. Tsuji and M. Sato and S. Yunoki and R. Wakizaka and Y. Kawashima and J. Doi and T. Itoko and H. Horii and T. Pellegrini and J. R. Moreno and K. J. Sung and E. Fejer and R. Walkup and S. Seelam and M. Motta},
  journal = {Preprint at 
  https://doi.org/10.48550/arXiv.2605.01138},
  year = {2026},
}

@article{Patel25_7490,
  author = {S. Patel and P. Jayakumar and T-C. Yen and A. F. Izmaylov},
  title = {{Quantum Measurement for Quantum Chemistry on a Quantum Computer}},
  journal = CP,
  volume = {126},
  number = {16},
  pages = {7490-7524},
  year = {2025},
}

@article{Yuri25_111335,
    author = {Y. Alexeev and V. S. Batista and N. Bauman and L. Bertels and D. Claudino and R. Dutta and L. Gagliardi and S. Godwin and N. Govind and M. Head-Gordon and M. R. Hermes and K. Kowalski and A. Li and C. Liu and J. Liu and P. Liu and J. M. Garc\'{i}a-Lastra and D. Mejia-Rodriguez and K. Mueller and M. Otten and B. Peng and M. Raugas and M. Reiher and P. Rigor and W. J. Shaw and M. van Schilfgaarde and T. Vegge and Y. Zhang and M. Zheng and L. Zhu},
    title = {{A Perspective on Quantum Computing Applications in Quantum Chemistry Using 25–100 Logical Qubits}},
    journal = JCTC,
    volume = {21},
    number = {22},
    pages = {11335-11357},
    year = {2025},
}

\onecolumngrid
\appendix
\section{Variance Bound of the Overlap Estimator}
\label{app:ov-variance}

We prove Theorem~\ref{thm:ov-variance} by evaluating the Born-weighted second moment of the estimator. Throughout this appendix, $1\leq\eta\leq N$, $N_{\mathrm{aug}}=N+\eta$, and the physical trial state $|\Psi_T\rangle$ is normalized and has fixed particle number $\eta$. We first take the target Slater determinant $|\phi\rangle$ to be normalized, with $\Phi^\dagger\Phi=I_\eta$. The measurement rotations are sampled independently from the complex Haar measure on $\mathcal U_{N_{\mathrm{aug}}}$, and occupation outcomes follow the ideal Born distribution of the augmented state $|\Psi_{\mathrm{aug}}\rangle=(|\Psi_T\rangle+|\vec a\rangle)/\sqrt{2}$. Here $\vec a=N+[\eta]$, $\Phi\in\mathbb C^{N\times\eta}$, and $|\Psi_T\rangle$ denotes its embedding with unoccupied ancillary modes, as above.

Retaining the overlap convention $o_\phi=\langle\phi|\Psi_T\rangle$, define the complex variance by
\begin{equation}
 \operatorname{Var}_{\mathbb C}(\widehat O_{\vec q,w})
 =\mathbb E_{u,\vec z}\!\left[
 |\widehat O_{\vec q,w}(u,\vec z)-o_\phi|^2\right].
 \label{eq:ov-var-definition}
\end{equation}
Let $p_j$ be the probability that $|\Psi_T\rangle$ has exactly $j$ particles outside the occupied one-particle subspace of $|\phi\rangle$. In the basis adapted to that subspace,
\begin{equation}
 p_j=\sum_{\substack{\vec r\in \SN{N}{\eta}\\
 |\vec r\setminus\vec q|=j}}
 |\langle\vec r|\Psi_T\rangle|^2,
 \qquad \sum_{j=0}^{\eta}p_j=1,
 \qquad p_0=|o_\phi|^2.
 \label{eq:ov-pj}
\end{equation}
Probabilities for sectors that do not exist are zero. Set
\begin{equation}
 b_{N,\eta}=(\eta+1)\frac{N+2}{N+\eta+2}.
 \label{eq:ov-b}
\end{equation}

We will establish the exact identities
\begin{align}
 \mathbb E_{u,\vec z}|\widehat O_{\vec q,w}|^2
 &=2b_{N,\eta}\left(1+\sum_{j=0}^{\eta}\frac{p_j}{j+1}\right),
 \label{eq:ov-exact-second}\\
 \operatorname{Var}_{\mathbb C}(\widehat O_{\vec q,w})
 &=2b_{N,\eta}\left(1+\sum_{j=0}^{\eta}\frac{p_j}{j+1}\right)-p_0.
 \label{eq:ov-exact-var}
\end{align}
These identities imply the tight bound in Theorem~\ref{thm:ov-variance}.

\subsection{Measurement Records and Born-Weighted Averaging}

We begin by connecting the matrix determinant estimator to the measurement distribution. Write
\begin{equation}
 d_j=\binom{N_{\mathrm{aug}}+j}{\eta},\qquad j=0,1,2.
 \label{eq:ov-dj}
\end{equation}
For a fixed record $(u,\vec z)$, define the back-rotated matrix determinant and its orbital frame as
\begin{equation}
 |V\rangle=U_\eta(u)^\dagger|\vec z\rangle,
 \qquad V=(u^\dagger)_{:,\vec z}\in\mathbb C^{N_{\mathrm{aug}}\times\eta},
 \qquad V^\dagger V=I_\eta.
 \label{eq:ov-frame}
\end{equation}
For $\vec r\in \SN{N_{\mathrm{aug}}}{\eta}$, let
\begin{equation}
 v_{\vec r}=\langle\vec r|V\rangle=\det V_{\vec r,:},
 \qquad x_{\vec r}=|v_{\vec r}|^2.
 \label{eq:ov-minors}
\end{equation}
The projector associated with this record is $|V\rangle\langle V|$. Since $\vec z$ is sampled after applying $U_\eta(u)$, its conditional probability is
\begin{equation}
 p(\vec z\mid u)
 =|\langle\vec z|U_\eta(u)|\Psi_{\mathrm{aug}}\rangle|^2
 =|\langle V|\Psi_{\mathrm{aug}}\rangle|^2.
 \label{eq:ov-born}
\end{equation}
Thus, for any integrable function $F$ of the orbital frame,
\begin{align}
 \mathbb E_{u,\vec z}F(V)
 &=\sum_{\vec z\in \SN{N_{\mathrm{aug}}}{\eta}}\mathbb E_u\!\left[
 p(\vec z\mid u)F\big((u^\dagger)_{:,\vec z}\big)\right]
 \notag\\
 &=d_0\,\mathbb E_V\!\left[|\langle V|\Psi_{\mathrm{aug}}\rangle|^2F(V)\right].
 \label{eq:ov-born-average}
\end{align}
On the right, $V$ is a uniformly distributed orthonormal frame. For every fixed $\vec z$, the selected columns $(u^\dagger)_{:,\vec z}$ have the same Haar distribution; the sum therefore contributes $d_0$ identical expectations. No independence between different outcomes of the same rotation is required. Equation~\eqref{eq:ov-born-average} distinguishes the uniform frame measure from the distribution of frames observed in actual measurement records. Its normalization follows by setting $F=1$.

A physical change of basis by $w^\dagger$ maps $|\phi\rangle=U_\eta(w)|\vec q\rangle$ to $|\vec q\rangle$ while leaving $|\vec a\rangle$ unchanged. The accompanying replacement $u\mapsto uW$, with $W=w\oplus I_\eta$, preserves Haar measure and all record probabilities after the state is transformed. We may therefore work in this adapted basis, relabel the transformed state and rotation as $|\Psi_T\rangle,u$, and set $w=I$ for the calculation. The patterns $\vec q$ and $\vec a$ remain disjoint.

The estimator previously derived from the inverse measurement channel can then be written as
\begin{align}
 \widehat O_{\vec q,I}(u,\vec z)
 &=2d_1\det\!\left[(U_o\Phi)^\dagger U_a\right]
 =2d_1v_{\vec q}v_{\vec a}^*.
 \label{eq:ov-estimator}
\end{align}
Indeed, in this basis $\Phi$ selects $\vec q$, so $(U_o\Phi)^\dagger U_a=V_{\vec q,:}V_{\vec a,:}^\dagger$. Equivalently, $v_{\vec q}v_{\vec a}^*=\operatorname{Tr}_{\wedge^\eta}[D_{\vec q}^{\vec a}|V\rangle\langle V|]$, with $D_{\vec q}^{\vec a}=|\vec a\rangle\langle\vec q|$. The factor $2$ recovers $o_\phi$ from the augmented coherence $\langle\Psi_{\mathrm{aug}}|D_{\vec q}^{\vec a}|\Psi_{\mathrm{aug}}\rangle=o_\phi/2$; the factor $d_1$ inverts the corresponding channel eigenvalue. Its normalization is also verified below.

\subsection{Diagonal Second-Moment Operator}

Define the positive operator
\begin{equation}
 \mathcal B_{\vec q,\vec a}
 =4d_0d_1^2\,\mathbb E_V\!\left[
 x_{\vec q}x_{\vec a}|V\rangle\langle V|\right].
 \label{eq:ov-second-operator}
\end{equation}
Equations~\eqref{eq:ov-born-average} and \eqref{eq:ov-estimator} imply
\begin{equation}
 \mathbb E_{u,\vec z}|\widehat O_{\vec q,I}|^2
 =\langle\Psi_{\mathrm{aug}}|\mathcal B_{\vec q,\vec a}|\Psi_{\mathrm{aug}}\rangle.
 \label{eq:ov-operator-expectation}
\end{equation}
We now show explicitly that $\mathcal B_{\vec q,\vec a}$ is diagonal in the occupation basis. For two patterns $\vec r,\vec s\in \SN{N_{\mathrm{aug}}}{\eta}$, its matrix element is
\begin{equation}
 \langle\vec r|\mathcal B_{\vec q,\vec a}|\vec s\rangle
 =4d_0d_1^2\,\mathbb E_V\!\left[
 x_{\vec q}(V)x_{\vec a}(V)v_{\vec r}(V)v_{\vec s}(V)^*\right].
 \label{eq:ov-offdiagonal-element}
\end{equation}
The explicit arguments emphasize that all four factors are functions of the same random frame. Consider the diagonal one-particle unitary
\begin{equation}
 R_{\boldsymbol\theta}
 =\operatorname{diag}(e^{i\theta_1},\ldots,e^{i\theta_{N_{\mathrm{aug}}}}),
 \qquad V\longmapsto R_{\boldsymbol\theta}V.
 \label{eq:ov-mode-phase}
\end{equation}
Each row $p$ of $V$ is multiplied by $e^{i\theta_p}$. Multilinearity of the matrix determinant therefore implies
\begin{equation}
 v_{\vec r}(R_{\boldsymbol\theta}V)
 =\det(R_{\boldsymbol\theta}V)_{\vec r,:}
 =e^{i\vartheta_{\vec r}}v_{\vec r}(V),
 \qquad \vartheta_{\vec r}=\sum_{p\in\vec r}\theta_p.
 \label{eq:ov-minor-phase}
\end{equation}
The squared magnitudes satisfy $x_{\vec r}(R_{\boldsymbol\theta}V)=x_{\vec r}(V)$. In particular, the weight $x_{\vec q}x_{\vec a}$ in Eq.~\eqref{eq:ov-offdiagonal-element} is unchanged even if the phases on the modes in $\vec q$ and $\vec a$ are varied independently.

Since $R_{\boldsymbol\theta}$ is unitary, $R_{\boldsymbol\theta}V$ has exactly the same Haar frame distribution as $V$. We may thus replace $V$ by $R_{\boldsymbol\theta}V$ inside the expectation in Eq.~\eqref{eq:ov-offdiagonal-element}. Using Eq.~\eqref{eq:ov-minor-phase} gives
\begin{align}
 \langle\vec r|\mathcal B_{\vec q,\vec a}|\vec s\rangle
 &=4d_0d_1^2\,\mathbb E_V\!\left[
 x_{\vec q}(V)x_{\vec a}(V)
 e^{i\vartheta_{\vec r}}v_{\vec r}(V)
 e^{-i\vartheta_{\vec s}}v_{\vec s}(V)^*\right]
 \notag\\
 &=e^{i(\vartheta_{\vec r}-\vartheta_{\vec s})}
 \langle\vec r|\mathcal B_{\vec q,\vec a}|\vec s\rangle.
 \label{eq:ov-phase-selection}
\end{align}
This identity holds for every choice of the $N_{\mathrm{aug}}$ mode phases. If $\vec r\ne\vec s$, their occupation patterns differ on at least one mode $p$. Choose $\theta_p=\pi$ on that mode and set all other phases to zero. Exactly one of $\vartheta_{\vec r}$ and $\vartheta_{\vec s}$ then contains $\pi$, so the phase factor in Eq.~\eqref{eq:ov-phase-selection} is $-1$. The matrix element must equal its own negative and hence vanish:
\begin{equation}
 \langle\vec r|\mathcal B_{\vec q,\vec a}|\vec s\rangle=0
 \quad\text{whenever }\vec r\ne\vec s.
 \label{eq:ov-offdiagonal-zero}
\end{equation}
For $\vec r=\vec s$, the phase factor is always one and imposes no such constraint. Thus
\begin{equation}
 \mathcal B_{\vec q,\vec a}
 =\sum_{\vec r\in \SN{N_{\mathrm{aug}}}{\eta}}
 \gamma_{\vec r}|\vec r\rangle\langle\vec r|.
 \label{eq:ov-diagonal-decomposition}
\end{equation}
Equivalently, the operator commutes with all independent mode-phase rotations, whose joint eigenvectors are the distinct occupation Slater determinants. A single common phase on all modes would not suffice: it acts identically on every state in the fixed-$\eta$ sector. This argument is a change of variables in the Haar expectation, not an additional operation in the measurement protocol. It also makes no assumption that the trial-state density matrix is diagonal; the trial state enters separately through Eq.~\eqref{eq:ov-operator-expectation}. The diagonal entries are
\begin{equation}
 \gamma_{\vec r}
 =\langle\vec r|\mathcal B_{\vec q,\vec a}|\vec r\rangle
 =4d_0d_1^2\,\mathbb E_V[x_{\vec q}x_{\vec a}x_{\vec r}].
 \label{eq:ov-eigenvalue-moment}
\end{equation}
Permutations within $\vec q$, within $\vec a$, and within their complement imply that this eigenvalue depends only on
\begin{equation}
 k=|\vec r\cap\vec q|,\qquad
 \ell=|\vec r\cap\vec a|,\qquad
 m=\eta-k-\ell.
 \label{eq:ov-sectors}
\end{equation}
Present sectors satisfy $k,\ell,m\geq0$ and $m\leq N_{\mathrm{aug}}-2\eta$.

\subsection{Gaussian Evaluation of the Three-Minor Moment}

The preceding subsection reduces the second-moment problem to a scalar moment calculation. Equation~\eqref{eq:ov-eigenvalue-moment} expresses each eigenvalue of $\mathcal B_{\vec q,\vec a}$ as $4d_0d_1^2$ times $\mathbb E_V[x_{\vec q}x_{\vec a}x_{\vec r}]$, and Eq.~\eqref{eq:ov-sectors} shows that the occupation pattern $\vec r$ enters only through the counts $k$, $\ell$, and $m$. Diagonality identifies the eigenvectors but does not determine these eigenvalues. We must therefore evaluate this three-minor moment for each allowed sector before finding the largest eigenvalue and the variance on $|\Psi_{\mathrm{aug}}\rangle$.

The orthonormality constraint on $V$ correlates its rows, so the three squared minors cannot be averaged independently. We evaluate their joint moment by introducing a standard complex-Gaussian matrix $\mathcal G=V\Xi^{1/2}$. Its orthonormal factor $V$ has the required Haar distribution and is independent of the positive factor $\Xi$. Each squared minor of $\mathcal G$ equals the corresponding $x_{\vec r}$ times $\det\Xi$. Thus the desired Haar moment is obtained by dividing a three-minor Gaussian moment by $\mathbb E(\det\Xi)^3$. The Gaussian calculation is simpler because disjoint row sets are independent, while the shared rows are specified exactly by $k$ and $\ell$.

We first establish the Gaussian Gram-determinant identities needed for this conversion and for integrating the unshared rows. We give the required part of the complex Bartlett decomposition explicitly, including the Gaussian normalization and the independence of its factors.

Let $F\in\mathbb C^{p\times h}$, $p\leq h$, have independent standard complex-Gaussian entries. Explicitly, each entry has the form
\begin{equation}
 F_{ij}=\frac{\xi_{ij}+i\zeta_{ij}}{\sqrt{2}},
 \qquad \xi_{ij},\zeta_{ij}\overset{\mathrm{iid}}{\sim}\mathcal N(0,1).
\end{equation}
Thus $\mathbb E F_{ij}=0$ and $\mathbb E|F_{ij}|^2=1$. The real and imaginary parts each have variance $1/2$. With this convention, $FF^\dagger$ is a complex Wishart matrix. Since the Gaussian distribution is continuous and $p\leq h$, $F$ has full row rank almost surely.

Write the columns of $F^\dagger$ as $f_1,\ldots,f_p\in\mathbb C^h$. Gram--Schmidt defines the residual vectors and their normalized directions by
\begin{equation}
 z_j=\left(I_h-\sum_{i=1}^{j-1}q_iq_i^\dagger\right)f_j,
 \qquad q_j=\frac{z_j}{\|z_j\|},
 \qquad j=1,\ldots,p,
\end{equation}
where the sum is empty at $j=1$. The resulting thin QR factorization is $F^\dagger=QR$, with $Q^\dagger Q=I_p$ and $R$ upper triangular. Choose its diagonal entries positive, so that $R_{jj}=\|z_j\|$. Then
\begin{equation}
 FF^\dagger=R^\dagger R,
 \qquad
 \det(FF^\dagger)=|\det R|^2
 =\prod_{j=1}^{p}\|z_j\|^2.
\end{equation}
Define $T_j=\|z_j\|^2$. To calculate the matrix determinant moment, we must determine the distributions of the $T_j$ and show that they are independent.

Condition on the preceding columns $f_1,\ldots,f_{j-1}$. Their span has dimension $j-1$ almost surely, and the projection defining $z_j$ has rank $h-j+1$. Choose an orthonormal coordinate system adapted to this span. The vector $f_j$ is independent of the preceding columns and its standard complex-Gaussian distribution is invariant under unitary rotations. In the adapted coordinates it therefore still has independent standard complex-Gaussian components. The projection removes the first $j-1$ components, leaving
\begin{equation}
 T_j\ \overset{d}{=}\ \sum_{\alpha=1}^{h-j+1}|g_\alpha|^2,
\end{equation}
where the $g_\alpha$ are independent standard complex-Gaussian scalars. For one such scalar $g$, the density on the complex plane is $\pi^{-1}e^{-|g|^2}$. Polar coordinates give, for $t\geq0$,
\begin{equation}
 \Pr(|g|^2>t)=\int_{\sqrt{t}}^\infty 2r e^{-r^2}\,dr=e^{-t}.
\end{equation}
Thus $|g|^2$ is exponential with unit mean, and the sum of $h-j+1$ such independent variables is gamma-distributed:
\begin{equation}
 T_j\mid(f_1,\ldots,f_{j-1})
 \sim\operatorname{Gamma}(h-j+1,1).
\end{equation}
Here the gamma distribution is parameterized by shape and scale. This conditional distribution does not depend on the values of the preceding columns. Hence $T_j$ is independent of those columns, and therefore of $T_1,\ldots,T_{j-1}$. The squared residual lengths are mutually independent; this is the diagonal part of the complex Bartlett decomposition.

If $T\sim\operatorname{Gamma}(s,1)$, then for an integer $a\geq0$,
\begin{equation}
 \mathbb E[T^a]
 =\frac{1}{\Gamma(s)}\int_0^\infty t^{s+a-1}e^{-t}\,dt
 =\frac{\Gamma(s+a)}{\Gamma(s)}
 =\prod_{b=0}^{a-1}(s+b).
\end{equation}
The symbol $\Gamma(s)$ in this expression denotes the Euler gamma function. Applying this identity with $s=h-j+1$ and using the independence just proved gives
\begin{equation}
 \mathbb E[\det(FF^\dagger)^a]
 =\mathbb E\!\left[\prod_{j=1}^{p}T_j^a\right]
 =\prod_{j=1}^{p}\mathbb E[T_j^a].
\end{equation}
Consequently,
\begin{equation}
 \mathbb E\!\left[\det(FF^\dagger)^a\right]
 =\prod_{j=1}^{p}\prod_{b=0}^{a-1}(h-j+1+b).
 \label{eq:ov-wishart}
\end{equation}
In particular,
\begin{equation}
 \mathbb E\det(FF^\dagger)=\frac{h!}{(h-p)!},\qquad
 \mathbb E\det(FF^\dagger)^2
 =\frac{h!(h+1)!}{(h-p)!(h-p+1)!}.
 \label{eq:ov-wishart-two}
\end{equation}
Empty matrix determinants and empty products are one.

Let $\mathcal G\in\mathbb C^{N_{\mathrm{aug}}\times\eta}$ be standard complex Gaussian, and use its polar decomposition
\begin{equation}
 \mathcal G=V\Xi^{1/2},\qquad \Xi=\mathcal G^\dagger\mathcal G.
 \label{eq:ov-polar}
\end{equation}
To justify the distribution and independence, observe that for every fixed
$R\in\mathcal U_{N_{\mathrm{aug}}}$, the matrices $R\mathcal G$ and
$\mathcal G$ have the same distribution. Left multiplication changes the
polar factors by $(V,\Xi)\mapsto(RV,\Xi)$. Thus the conditional law of
$V$ given $\Xi$ is invariant under every left unitary rotation. The unique
invariant probability measure on orthonormal frames is Haar measure, so
this conditional law does not depend on $\Xi$. Hence $V$ is Haar-distributed
and independent of $\Xi$. For every occupation pattern $\vec r$,
\begin{equation}
 |\det\mathcal G_{\vec r,:}|^2=x_{\vec r}\det\Xi.
 \label{eq:ov-radial-minor}
\end{equation}
Equation~\eqref{eq:ov-wishart} gives
\begin{align}
 \mathbb E(\det\Xi)^3
 &=\prod_{j=1}^{\eta}(N_{\mathrm{aug}}-j+1)(N_{\mathrm{aug}}-j+2)(N_{\mathrm{aug}}-j+3)
 \notag\\
 &=(\eta!)^3d_0d_1d_2.
 \label{eq:ov-radial-third}
\end{align}
Thus, writing $\mathcal G_{\vec r}=\mathcal G_{\vec r,:}$,
\begin{equation}
 \mathbb E_V[x_{\vec q}x_{\vec a}x_{\vec r}]
 =\frac{\mathbb E[|\det\mathcal G_{\vec q}|^2
 |\det\mathcal G_{\vec a}|^2|\det\mathcal G_{\vec r}|^2]}
 {(\eta!)^3d_0d_1d_2}.
 \label{eq:ov-gaussian-conversion}
\end{equation}

To evaluate the numerator, let $F_q$, $F_a$, and $F_0$ contain the rows in $\vec r\cap\vec q$, $\vec r\cap\vec a$, and $\vec r\setminus(\vec q\cup\vec a)$, respectively. They have $k$, $\ell$, and $m$ rows and $\eta$ columns. Let $R_q$ and $R_a$ contain the remaining rows of $\vec q$ and $\vec a$, respectively. These five matrices are independent. Up to row permutations,
\begin{equation}
 \mathcal G_{\vec q}=\begin{bmatrix}F_q\\R_q\end{bmatrix},\qquad
 \mathcal G_{\vec a}=\begin{bmatrix}F_a\\R_a\end{bmatrix},\qquad
 \mathcal G_{\vec r}=\begin{bmatrix}F_q\\F_a\\F_0\end{bmatrix}.
 \label{eq:ov-row-partition}
\end{equation}
Row-permutation signs disappear in squared matrix determinant magnitudes.

For a fixed full-row-rank $F\in\mathbb C^{p\times\eta}$ and an independent standard Gaussian $Z\in\mathbb C^{(\eta-p)\times\eta}$,
\begin{equation}
 \mathbb E_Z\left|\det\begin{bmatrix}F\\Z\end{bmatrix}\right|^2
 =(\eta-p)!\det(FF^\dagger).
 \label{eq:ov-conditional-det}
\end{equation}
The expectation in Eq.~\eqref{eq:ov-conditional-det} is over $Z$ only; the rows of $F$ are held fixed. We derive it by isolating the part of the Gaussian rows perpendicular to the row space of $F$. Put $\nu=\eta-p$. Choose $Q_\parallel\in\mathbb C^{\eta\times p}$ with orthonormal columns spanning the range of $F^\dagger$, and $Q_\perp\in\mathbb C^{\eta\times\nu}$ with orthonormal columns spanning the kernel of $F$. Then
\begin{equation}
 Q=[Q_\parallel\;Q_\perp]\in\mathbb C^{\eta\times\eta},
 \qquad Q^\dagger Q=QQ^\dagger=I_\eta,
 \qquad FQ=[L\;0],\quad L=FQ_\parallel.
\end{equation}
The $p\times p$ matrix $L$ is nonsingular because $F$ has full row rank. This change of coordinates preserves the fixed Gram matrix:
\begin{equation}
 FF^\dagger=FQQ^\dagger F^\dagger=LL^\dagger,
 \qquad |\det L|^2=\det(FF^\dagger).
\end{equation}
For the random rows, write
\begin{equation}
 ZQ=[Z_\parallel\;Z_\perp],
 \qquad Z_\parallel=ZQ_\parallel\in\mathbb C^{\nu\times p},
 \qquad Z_\perp=ZQ_\perp\in\mathbb C^{\nu\times\nu}.
\end{equation}
Since $F$ is fixed, $Q$ is fixed when averaging over $Z$. Each row of $Z$ has density proportional to $e^{-\|z\|^2}$, invariant under right multiplication by a unitary. Therefore $ZQ$ has the same joint distribution as $Z$: its entries remain independent standard complex Gaussians. In particular, $Z_\perp$ is a square standard complex-Gaussian matrix. Multiplying the stacked matrix by $Q$ gives a block-triangular matrix,
\begin{equation}
 \begin{bmatrix}F\\Z\end{bmatrix}Q
 =\begin{bmatrix}L&0\\Z_\parallel&Z_\perp\end{bmatrix}.
\end{equation}
As $|\det Q|=1$, its squared matrix determinant satisfies the pointwise identity
\begin{equation}
 \left|\det\begin{bmatrix}F\\Z\end{bmatrix}\right|^2
 =|\det L|^2|\det Z_\perp|^2
 =\det(FF^\dagger)|\det Z_\perp|^2.
\end{equation}
Thus the random components $Z_\parallel$ parallel to the fixed row space make no contribution to this matrix determinant. The only remaining average is the squared matrix determinant of the $\nu\times\nu$ matrix $Z_\perp$.

This average equals $\nu!$, either by Eq.~\eqref{eq:ov-wishart-two} or directly by expanding the matrix determinant. To verify it without another matrix identity, let $\mathfrak S_\nu$ denote the permutation group and write
\begin{equation}
 \det Z_\perp
 =\sum_{\sigma\in\mathfrak S_\nu}\operatorname{sgn}(\sigma)
 \prod_{i=1}^{\nu}(Z_\perp)_{i,\sigma(i)}.
\end{equation}
Independence between rows and unit covariance within each row give
\begin{equation}
 \begin{aligned}
 \mathbb E|\det Z_\perp|^2
 &=\sum_{\sigma,\tau\in\mathfrak S_\nu}
 \operatorname{sgn}(\sigma)\operatorname{sgn}(\tau)
 \prod_{i=1}^{\nu}
 \mathbb E\!\left[(Z_\perp)_{i,\sigma(i)}
 (Z_\perp)_{i,\tau(i)}^*\right]\\
 &=\sum_{\sigma,\tau\in\mathfrak S_\nu}
 \operatorname{sgn}(\sigma)\operatorname{sgn}(\tau)
 \prod_{i=1}^{\nu}\delta_{\sigma(i),\tau(i)}
 =\sum_{\sigma\in\mathfrak S_\nu}1=\nu!.
 \end{aligned}
\end{equation}
Only equal permutations survive the average. Taking the expectation of the pointwise matrix determinant identity therefore yields
\begin{equation}
 \mathbb E_Z\left|\det\begin{bmatrix}F\\Z\end{bmatrix}\right|^2
 =\det(FF^\dagger)\,\nu!
 =(\eta-p)!\det(FF^\dagger),
\end{equation}
proving Eq.~\eqref{eq:ov-conditional-det}. Empty-block conventions cover $p=0$ and $p=\eta$. If $F$ is rank deficient, both sides vanish. The same result holds conditionally when $F$ is random but independent of $Z$: first fix $F$ and choose $Q$ from that fixed realization, then perform the Gaussian average. This is the form used below to integrate the unshared rows.

Define
\begin{equation}
 \Delta_q=\det(F_qF_q^\dagger),\quad
 \Delta_a=\det(F_aF_a^\dagger),\quad
 \Delta_{qa}=\det\!\left(
 \begin{bmatrix}F_q\\F_a\end{bmatrix}
 \begin{bmatrix}F_q\\F_a\end{bmatrix}^{\!\dagger}\right).
 \label{eq:ov-gram-dets}
\end{equation}
Condition on $F_q,F_a$ and integrate $R_q,R_a,F_0$ independently using Eq.~\eqref{eq:ov-conditional-det}. The numerator of Eq.~\eqref{eq:ov-gaussian-conversion} becomes
\begin{equation}
 (\eta-k)!(\eta-\ell)!m!\,
 \mathbb E[\Delta_q\Delta_a\Delta_{qa}].
 \label{eq:ov-conditioned-product}
\end{equation}

Condition on $F_q$ and choose a unitary column rotation $Q$ such that
$F_qQ=[L\;0]$, with $L\in\mathbb C^{k\times k}$ nonsingular almost surely.
Write $F_aQ=[A\;B]=[F_{a,\parallel}\;F_{a,\perp}]$, where
$B\in\mathbb C^{\ell\times(\eta-k)}$. The stacked Gram matrix is
\begin{equation}
 \begin{bmatrix}
 LL^\dagger & LA^\dagger\\
 AL^\dagger & AA^\dagger+BB^\dagger
 \end{bmatrix}.
\end{equation}
Multiplying this matrix on the left by the determinant-one matrix
\begin{equation}
 \begin{bmatrix}
 I_k&0\\-AL^\dagger(LL^\dagger)^{-1}&I_\ell
 \end{bmatrix}
\end{equation}
eliminates its lower-left block. Since
$L^\dagger(LL^\dagger)^{-1}L=I_k$, the lower-right block becomes
$AA^\dagger+BB^\dagger-AA^\dagger=BB^\dagger$. Taking matrix determinants gives
\begin{equation}
 \Delta_{qa}=\Delta_q\det(F_{a,\perp}F_{a,\perp}^\dagger).
 \label{eq:ov-schur}
\end{equation}
The same identity holds for $k=0$ or $\ell=0$ with empty matrix determinants
interpreted as one.

The rotation $Q$ depends on $F_q$, so factorization of the expectation
requires conditioning. Given $F_q$, the matrix $Q$ is fixed and $F_a$
is still standard complex Gaussian. Right-unitary invariance implies that
$Z=F_aQ$, conditional on $F_q$, has the standard complex-Gaussian law,
independent of the value of $F_q$. Partition its columns as
$Z=[Z_\parallel\;Z_\perp]$, with $\eta-k$ columns in $Z_\perp$, and define
\begin{equation}
 C_{k,\ell}=\mathbb E_Z\!\left[
 \det(ZZ^\dagger)\det(Z_\perp Z_\perp^\dagger)\right].
\end{equation}
Because $F_aF_a^\dagger=ZZ^\dagger$, the conditional expectation of
$\Delta_a\det(F_{a,\perp}F_{a,\perp}^\dagger)$ given $F_q$ equals
$C_{k,\ell}$. In particular it does not depend on $F_q$. Therefore
\begin{equation}
 \begin{aligned}
 \mathbb E[\Delta_q\Delta_a\Delta_{qa}]
 &=\mathbb E_{F_q}\!\left[\Delta_q^2
 \mathbb E\!\left[
 \Delta_a\det(F_{a,\perp}F_{a,\perp}^\dagger)\mid F_q
 \right]\right]\\
 &=\mathbb E_{F_q}[\Delta_q^2]C_{k,\ell}.
 \end{aligned}
\end{equation}
Equivalently,
\begin{equation}
 \mathbb E[\Delta_q\Delta_a\Delta_{qa}]
 =\mathbb E[\Delta_q^2]\,
 \mathbb E[\Delta_a\det(F_{a,\perp}F_{a,\perp}^\dagger)].
 \label{eq:ov-factor-moment}
\end{equation}
The two matrix determinants inside the second expectation share Gaussian entries
and are not factored further.
The first factor is
\begin{equation}
 \mathbb E[\Delta_q^2]
 =\frac{\eta!(\eta+1)!}{(\eta-k)!(\eta-k+1)!}.
 \label{eq:ov-first-factor}
\end{equation}
For the second factor, use $Z=F_aQ$ as above. Cauchy--Binet gives
\begin{equation}
 \begin{aligned}
 \Delta_a=\det(ZZ^\dagger)
 &=\sum_{\substack{J\subseteq[\eta]\\|J|=\ell}}
 \det Z_{:,J}\det(Z_{:,J}^\dagger)\\
 &=\sum_{\substack{J\subseteq[\eta]\\|J|=\ell}}
 |\det Z_{:,J}|^2.
 \end{aligned}
 \label{eq:ov-cauchy-binet}
\end{equation}
There are $\binom{\eta}{\ell}$ summands. Multiplying by $\Delta_a$ and
averaging yields
\begin{equation}
 \mathbb E[\Delta_a^2]
 =\sum_{\substack{J\subseteq[\eta]\\|J|=\ell}}
 \mathbb E[\Delta_a|\det Z_{:,J}|^2].
\end{equation}
Column permutations preserve both the distribution of $Z$ and
$\Delta_a=\det(ZZ^\dagger)$, and map any $\ell$-column subset to any other.
Every term in this last sum therefore has the same expectation:
\begin{equation}
 \mathbb E[\Delta_a|\det Z_{:,J}|^2]
 =\frac{\mathbb E[\Delta_a^2]}{\binom{\eta}{\ell}}.
\end{equation}
Let $P=\{k+1,\ldots,\eta\}$ denote the columns of $Z_\perp$.
A second application of Cauchy--Binet gives
\begin{equation}
 \det(Z_\perp Z_\perp^\dagger)
 =\sum_{\substack{J\subseteq P\\|J|=\ell}}|\det Z_{:,J}|^2.
\end{equation}
Only $\binom{\eta-k}{\ell}$ subsets occur in this sum. Using the equal-term
expectation and then Eq.~\eqref{eq:ov-wishart-two}, we obtain
\begin{align}
 &\mathbb E[\Delta_a\det(F_{a,\perp}F_{a,\perp}^\dagger)]
 \notag\\
 &\quad=\frac{\binom{\eta-k}{\ell}}{\binom{\eta}{\ell}}
 \mathbb E[\Delta_a^2]
 \notag\\
 &\quad=\frac{(\eta-k)!(\eta-\ell)!}{m!\,\eta!}
 \frac{\eta!(\eta+1)!}{(\eta-\ell)!(\eta-\ell+1)!}
 \notag\\
 &\quad=\frac{(\eta-k)!(\eta+1)!}{m!(\eta-\ell+1)!}.
 \label{eq:ov-second-factor}
\end{align}
Here $m=\eta-k-\ell$. The binomial ratio follows from column symmetry,
without assuming independence between a minor and the full Gram matrix determinant.
Combining Eqs.~\eqref{eq:ov-conditioned-product}, \eqref{eq:ov-factor-moment}, \eqref{eq:ov-first-factor}, and \eqref{eq:ov-second-factor} and canceling factorials yields
\begin{equation}
 \mathbb E[|\det\mathcal G_{\vec q}|^2
 |\det\mathcal G_{\vec a}|^2|\det\mathcal G_{\vec r}|^2]
 =\frac{\eta![(\eta+1)!]^2}{(\eta-k+1)(\eta-\ell+1)}.
 \label{eq:ov-gaussian-third}
\end{equation}
Therefore,
\begin{equation}
 \mathbb E_V[x_{\vec q}x_{\vec a}x_{\vec r}]
 =\frac{(\eta+1)^2}
 {(\eta-k+1)(\eta-\ell+1)d_0d_1d_2}.
 \label{eq:ov-haar-third}
\end{equation}

For completeness, the same polar decomposition verifies the estimator normalization. Disjointness makes $\mathcal G_{\vec q}$ and $\mathcal G_{\vec a}$ independent, so their product of squared matrix determinant magnitudes has mean $(\eta!)^2$. Since $\mathbb E(\det\Xi)^2=(\eta!)^2d_0d_1$,
\begin{equation}
 \mathbb E_V[x_{\vec q}x_{\vec a}]=\frac1{d_0d_1}.
 \label{eq:ov-haar-second}
\end{equation}
Expanding $|\langle V|\Psi_{\mathrm{aug}}\rangle|^2$ in Eq.~\eqref{eq:ov-born-average}, independent mode phases leave only the term $\langle\vec q|\Psi_T\rangle x_{\vec q}x_{\vec a}/2$ in the mean of $v_{\vec q}v_{\vec a}^*$. Thus
\begin{equation}
 \mathbb E_{u,\vec z}[v_{\vec q}v_{\vec a}^*]
 =\frac{d_0}{2}\langle\vec q|\Psi_T\rangle
 \mathbb E_V[x_{\vec q}x_{\vec a}]
 =\frac{o_\phi}{2d_1},
 \label{eq:ov-unbiased-check}
\end{equation}
which confirms Eq.~\eqref{eq:ov-estimator} without an additional channel calculation.

\subsection{Exact spectrum and tight variance}

Substitution of Eq.~\eqref{eq:ov-haar-third} into Eq.~\eqref{eq:ov-eigenvalue-moment} gives every sector eigenvalue:
\begin{equation}
 \gamma_{k,\ell}
 =4(\eta+1)^2\frac{N_{\mathrm{aug}}-\eta+2}{N_{\mathrm{aug}}+2}
 \frac1{(\eta-k+1)(\eta-\ell+1)}.
 \label{eq:ov-full-spectrum}
\end{equation}
The denominator satisfies
\begin{equation}
 (\eta-k+1)(\eta-\ell+1)
 =(\eta+1)(m+1)+k\ell\geq\eta+1.
 \label{eq:ov-edge-inequality}
\end{equation}
Equality requires $m=0$ and $k\ell=0$, corresponding precisely to $\vec r=\vec q$ or $\vec r=\vec a$. Consequently,
\begin{equation}
 \|\mathcal B_{\vec q,\vec a}\|_\infty
 =\gamma_{\eta,0}=\gamma_{0,\eta}
 =4(\eta+1)\frac{N_{\mathrm{aug}}-\eta+2}{N_{\mathrm{aug}}+2}=4b_{N,\eta}.
 \label{eq:ov-operator-norm}
\end{equation}

\begin{proof}[Proof of Theorem~\ref{thm:ov-variance}]
The physical branch of $|\Psi_{\mathrm{aug}}\rangle$ has no occupation in $\vec a$. A physical configuration with $j$ particles outside $\vec q$ lies in the sector $k=\eta-j$, $\ell=0$, whose eigenvalue is $4b_{N,\eta}/(j+1)$. The ancillary branch has eigenvalue $4b_{N,\eta}$. Since the second-moment operator is diagonal, their coherence contributes zero, and
\begin{align}
 \langle\Psi_{\mathrm{aug}}|\mathcal B_{\vec q,\vec a}|\Psi_{\mathrm{aug}}\rangle
 &=\frac12\langle\Psi_T|\mathcal B_{\vec q,\vec a}|\Psi_T\rangle
 +\frac12\gamma_{0,\eta}
 \notag\\
 &=2b_{N,\eta}\left(1+\sum_{j=0}^{\eta}\frac{p_j}{j+1}\right).
 \label{eq:ov-physical-second}
\end{align}
Unbiasedness gives the variance by subtracting $|o_\phi|^2=p_0$. Moreover,
\begin{equation}
 \sum_{j=0}^{\eta}\frac{p_j}{j+1}
 \leq p_0+\frac{1-p_0}{2}=\frac{1+p_0}{2}.
 \label{eq:ov-p-bound}
\end{equation}
Since $b_{N,\eta}>1$,
\begin{equation}
 \operatorname{Var}_{\mathbb C}(\widehat O_{\vec q,w})
 \leq3b_{N,\eta}+(b_{N,\eta}-1)p_0
 \leq4b_{N,\eta}-1.
 \label{eq:ov-var-max}
\end{equation}
All inequalities are equalities for $p_0=1$, establishing tightness. Finally, $(N+2)/(N+\eta+2)\leq1$ gives the mode-independent bound.
\end{proof}

The second-moment bound and variance bound are distinct: the former is $4b_{N,\eta}$, whereas the latter is $4b_{N,\eta}-1$. The variance is also uniformly bounded below for this augmented-state protocol. Positivity and the fixed reference weight imply
\begin{equation}
 \mathbb E|\widehat O_{\vec q,w}|^2\geq2b_{N,\eta},\qquad
 \operatorname{Var}_{\mathbb C}(\widehat O_{\vec q,w})
 \geq2b_{N,\eta}-1\geq\eta+1,
 \label{eq:ov-lower-bound}
\end{equation}
since $N\geq\eta$ implies $b_{N,\eta}\geq(\eta+2)/2$. Thus the single-shadow complex variance is $\Theta(\eta)$ uniformly over normalized trial-Slater-determinant pairs and allowed mode counts. A spectral average over the entire augmented particle sector does not describe this restricted family of inputs, each of which has probability $1/2$ on $|\vec a\rangle$.

\subsection{Sample complexity and general walkers}

\paragraph{Fixed-target overlap estimation.}
For $N_{\mathrm{shadow}}$ independent shadow records, the ordinary mean satisfies
\begin{equation}
 \mathbb E|\widehat o_\phi-o_\phi|^2
 =\frac{\operatorname{Var}_{\mathbb C}(\widehat O_{\vec q,w})}{N_{\mathrm{shadow}}}
 \leq\frac{4\eta+3}{N_{\mathrm{shadow}}}.
 \label{eq:ov-mean-mse}
\end{equation}
For $L$ fixed normalized target Slater determinants, coordinatewise median-of-means estimates achieve complex absolute error at most $\epsilon$ simultaneously, with probability at least $1-\delta$, using
\begin{equation}
 N_{\mathrm{shadow}}=\mathcal{O}\!\left(\frac{\eta}{\epsilon^2}\log\frac{2L}{\delta}\right)
 \label{eq:ov-sample-complexity}
\end{equation}
independent records shared by all targets.
\par

For a general full-column-rank walker matrix $\Phi$, let $\Phi=QR$ be its thin QR factorization. The Slater determinant state and estimator both acquire the bra-oriented factor $\det(R)^*$, so
\begin{equation}
 \operatorname{Var}_{\mathbb C}(\widehat O_{\Phi})
 =\det(\Phi^\dagger\Phi)\,
 \operatorname{Var}_{\mathbb C}(\widehat O_Q)
 \leq\det(\Phi^\dagger\Phi)(4b_{N,\eta}-1).
 \label{eq:ov-general-walker}
\end{equation}
Here $\widehat O_\Phi$ denotes the same matrix determinant estimator evaluated with the orbital matrix $\Phi$. If $\Phi$ is rank deficient, both the Slater determinant state and estimator vanish. The conjugate estimator $\widehat D_{u,\vec z}(\Phi)$ has the same complex variance. These results establish additive-error guarantees for fixed overlap targets; relative errors near overlap zeros and nonlinear ratios along shadow-dependent AFQMC trajectories require separate analysis.

\section{Virtual Correlation Energy}
As shown in the following, orbital-rotation shadows also lead to a simpler VCE formulation.
Assuming an unrestricted framework, we first write the trial wave function as
\begin{equation}
    \ket{\Psi_T}=\ket{\Xi^\alpha_c}\otimes\ket{\Xi^\beta_c}\otimes\ket{\Psi_{T,a}}\otimes\ket{0^\alpha_v}\otimes\ket{0^\beta_v},
\end{equation}
where $\ket{\Psi_{T,a}}$ is the trial state prepared on a quantum computer within the active space with $\eta$ 
electrons, $\ket{\Xi^{\alpha (\beta)}_c}$ is a Slater determinant containing all frozen occupied $\alpha$ ($\beta$) orbitals with $\eta^{\alpha (\beta)}_c$ electrons, and $\ket{0^{\alpha (\beta)}_v}$ is the vacuum state in the space 
of virtual $\alpha$ ($\beta$) orbitals, with $\eta=\eta_a^\alpha+\eta_a^\beta$. The trial wave function can always be written as
\begin{equation}
    \ket{\Psi_T}=\ket{\Xi^\alpha_c}\otimes\ket{\Xi^\beta_c}\otimes\sum_{i}{c_i\ket{\chi_i^\alpha}\otimes\ket{\chi_i^\beta}}\otimes\ket{0^\alpha_v}\otimes\ket{0^\beta_v},
\end{equation}
where $\ket{\chi_i^{\alpha (\beta)}}$ is the $\alpha$ ($\beta$) component of the $i$th Slater determinant within the
active space. 

The overlap between the trial state and a walker defined in the full space is 
\begin{equation}
    \sum_i{c^\ast_i\braket{\Xi_c^\alpha\Xi_c^\beta\chi^\alpha_i\chi^\beta_i0^\alpha_v0^\beta_v}{\phi}}=\sum_i{c^\ast_i\mathrm{det}\left(\begin{pmatrix}
\Xi_c^\alpha & 0 & 0 & 0\\
0 & \Xi_c^\beta & 0 & 0 \\
0 & 0 & \chi^\alpha_i & 0 \\
0 & 0 & 0 & \chi^\beta_i \\
0 & 0 & 0 & 0 \\
0 & 0 & 0 & 0 \\
\end{pmatrix}^\dagger\begin{pmatrix}
\phi^\alpha_c & 0 \\
0 & \phi^\beta_c  \\
\phi^\alpha_a & 0 \\
0 & \phi^\beta_a \\
\phi^\alpha_v & 0  \\
0 & \phi^\beta_v  \\
\end{pmatrix}\right)}
\end{equation}
This becomes
\begin{equation}
    \sum_i{c^\ast_i\braket{\Xi_c^\alpha\Xi_c^\beta\chi^\alpha_i\chi^\beta_i}{\phi}}=\sum_i{c^\ast_i\mathrm{det}\begin{pmatrix}
\Xi_c^{\alpha\dagger}\phi^\alpha_c & 0 \\
0 & \Xi_c^{\beta\dagger}\phi^\beta_c \\
\chi^{\alpha\dagger}_i\phi^\alpha_a & 0 \\
0 & \chi^{\beta\dagger}_i\phi^\beta_a  \\
\end{pmatrix}}
\end{equation}
where $\phi^{\alpha (\beta)}_c$ and $\phi^{\alpha (\beta)}_a$ are $\eta^{\alpha(\beta)}_a+\eta^{\alpha(\beta)}_c$ column
molecular orbital coefficients of the walker. $\Xi_c^{\alpha(\beta)}$ is diagonal with ones up to the number of $\alpha$($\beta$) core
electrons and zeros elsewhere. The virtual degrees of freedom therefore no longer appear.

To further remove the core degrees of freedom, we first define
\begin{equation}
    M_c^{\alpha(\beta)}=\Xi_c^{\alpha(\beta)\dagger}\phi^{\alpha(\beta)}_c
\end{equation}
and perform the QR decomposition
\begin{equation}
    \label{eqn:vce_svd}
    \begin{split}
        M_c^{\alpha\dagger}=Q_c^\alpha R_c^\alpha \\
        M_c^{\beta\dagger}=Q_c^\beta R_c^\beta, \\
    \end{split}
\end{equation}
where $Q^\alpha_c\in\mathbb{C}^{(\eta^\alpha_c+\eta^\alpha_a)\times (\eta^\alpha_c+\eta^\alpha_a)}$ and $Q^\beta_c\in\mathbb{C}^{(\eta^\beta_c+\eta^\beta_a)\times (\eta^\beta_c+\eta^\beta_a)}$ are unitary matrices. These matrices can be partitioned as
\begin{equation}
    Q_c^{\alpha(\beta)}=\begin{pmatrix}
        Q_{1c}^{\alpha(\beta)} & Q_{2c}^{\alpha(\beta)}
    \end{pmatrix}
\end{equation}
where $Q_{1c}^{\alpha(\beta)}\in\mathbb{C}^{(\eta^{\alpha(\beta)}_c+\eta^{\alpha(\beta)}_a)\times \eta^{\alpha(\beta)}_c}$ and $Q_{2c}^{\alpha(\beta)}\in\mathbb{C}^{(\eta^{\alpha(\beta)}_c+\eta^{\alpha(\beta)}_a)\times \eta^{\alpha(\beta)}_a}$.
The $R$ matrix is 
\begin{equation}
    R^{\alpha(\beta)}_c=\begin{pmatrix}
        \tilde{R}^{\alpha(\beta)} \\
        0 \\
    \end{pmatrix}
\end{equation}
where $\tilde{R}^{\alpha(\beta)}$ is an $\eta_c^{\alpha(\beta)}\times \eta_c^{\alpha(\beta)}$ upper-triangular matrix. Defining $L_c^{\alpha(\beta)}=(\tilde{R}^{\alpha(\beta)})^\dagger$, we have

\begin{equation}
    \begin{split}
        M_c^{\alpha}=L_c^\alpha Q_{1c}^{\alpha\dagger}  \\
        M_c^{\beta}= L_c^\beta Q_{1c}^{\beta\dagger} . \\
    \end{split}
\end{equation}
Then, if we define
\begin{equation}
    Q=\begin{pmatrix}
Q^\alpha_{1c} & 0 & Q^{\alpha}_{2c} & 0 \\
0 & Q^\beta_{1c} & 0 & Q^{\beta}_{2c} \\
\end{pmatrix},
\end{equation}
the overlap can be written as
\begin{equation}
    \begin{split}
        \braket{\Xi_c^\alpha\Xi_c^\beta\chi^\alpha_i\chi^\beta_i}{\phi}&=\mathrm{det}\left(\begin{pmatrix}
\Xi_c^\alpha & 0 & 0 & 0\\
0 & \Xi_c^\beta & 0 & 0 \\
0 & 0 & \chi^\alpha_i & 0 \\
0 & 0 & 0 & \chi^\beta_i \\
\end{pmatrix}^\dagger\begin{pmatrix}
\phi^\alpha_c & 0 \\
0 & \phi^\beta_c  \\
\phi^\alpha_a & 0 \\
0 & \phi^\beta_a  \\
\end{pmatrix}Q\right) / \mathrm{det}(Q) \\
&=\mathrm{det}(L^\alpha_c)\mathrm{det}(L^\beta_c)\mathrm{det}(\chi^{\alpha\dagger}_i\tilde{\phi}^{\alpha}_a)\mathrm{det}(\chi^{\beta\dagger}_i\tilde{\phi}^{\beta}_a) / \mathrm{det}(Q), \\
    \end{split}
\end{equation}
where $\tilde{\phi}^{\alpha(\beta)}_a=\phi^{\alpha(\beta)}_aQ^{\alpha(\beta)}_{2c}$.

Finally, computing the overlap between the trial and Slater determinant in the full
space reduces to a core-space prefactor times an active-space overlap:
\begin{equation}
  \braket{\Psi_T}{\phi}
  =
  \frac{
  \det(L^\alpha_c)\det(L^\beta_c)
  }{
  \det(Q)
  }
  \braket{\Psi_{T,a}}{\tilde{\phi}_a}.
\end{equation}
This factorization separates a VCE prefactor (from QR over core/active blocks)
from an active-space overlap evaluated with orbital rotation shadows.
In practice, analytical VCE mainly accelerates derivatives of the prefactor and
construction of derivative tensors (e.g., $dQ$, $d^2Q$), and this benefit grows
with core-space size.
By contrast, derivatives of $\braket{\Psi_{T,a}}{\tilde{\phi}_a}$ scale with the
number of shadow samples $N_{\mathrm{shadow}}$ and can dominate when $N_{\mathrm{shadow}}$
is large.

\end{document}